\documentclass[journal]{IEEEtran}
\ifCLASSINFOpdf
\else
\fi
\usepackage{booktabs}
\usepackage{graphicx} 
\usepackage{stfloats}

\usepackage{textcomp}
\usepackage{stfloats}
\usepackage{verbatim}
\usepackage{graphicx}
\usepackage{booktabs}
\usepackage{amsmath,amsfonts}
\usepackage{algorithmic}
\usepackage{array}
\usepackage{stfloats}
\usepackage{verbatim}
\usepackage{graphicx}
\usepackage{balance}
\usepackage{balance}
\usepackage[utf8]{inputenc}
\usepackage{caption}
\usepackage{graphicx}
\usepackage{float} 
\usepackage{subcaption}
\usepackage{color}
\usepackage{makecell}
\usepackage{orcidlink}
\usepackage{amsthm}
\usepackage{amssymb}

\newtheorem{theorem}{Theorem}
\newtheorem{corollary}{Corollary}

\begin{document}
\title{\textcolor{black}{Diffusion Models for Radio Map Estimation: Theoretical Performance Analysis and\\Sampling Rate Guideline}}
\author{Zhiyuan~Liu\orcidlink{0009-0003-1890-5115},~\IEEEmembership{Student Member,~IEEE, }
Qingyu~Liu\orcidlink{0000-0003-1487-7682},~\IEEEmembership{Member,~IEEE, }
Shuhang~Zhang\orcidlink{0000-0003-1313-2726},~\IEEEmembership{Member,~IEEE, }
Hongliang~Zhang\orcidlink{0000-0003-3393-8612},~\IEEEmembership{Member,~IEEE, } Lingyang~Song\orcidlink{0000-0001-8644-8241},~\IEEEmembership{Fellow,~IEEE}

\thanks{Received 28 February 2026; revised 27 June 2026; accepted 24 October 2026.
This work was supported in part by the National Natural Science Foundation of China under Grant 62527801, 62401025, and 62401302; in part by the Key Project of Beijing Natural Science Foundation under Grant L253004.
\emph{(Corresponding authors: Qingyu Liu; Lingyang Song)}
}
\thanks{Zhiyuan Liu and Qingyu Liu are with the School of Electronic and Computer Engineering,  Peking University Shenzhen Graduate School, Shenzhen, China, 518055; and with the Peng Cheng Laboratory, Shenzhen, China, 518055 (e-mail: 
{liuzhiyuan@stu.pku.edu.cn}, {qy.liu@pku.edu.cn}).}
\thanks{Shuhang Zhang and Hongliang Zhang are with the State Key Laboratory of Photonics and Communications, School of Electronics, Peking University, Beijing, China, 100871 (email: 
{shuhangzhang@pku.edu.cn}, {hongliang.zhang@pku.edu.cn}).}
\thanks{Lingyang Song is with the State Key Laboratory of Photonics and Communications, School of Electronics, Peking University, Beijing, China, 100871; with the School of Electronic and Computer Engineering, Peking University Shenzhen Graduate School, Shenzhen, China, 518055; and with the Peng Cheng Laboratory, Shenzhen, China, 518055 (e-mail: 
{lingyang.song@pku.edu.cn}).}
}

\maketitle
\thispagestyle{empty}
\pagestyle{empty}

\begin{abstract}
Radio maps, which characterize the spatial distribution of radio frequency metrics, such as the received signal strength, are essential for a wide range of wireless applications.
The problem of radio map estimation involves constructing a radio map from \textcolor{black}{a limited set of radio samples measured by sparsely distributed sensors.}
Recently, diffusion models have been increasingly adopted for this problem, yet \textcolor{black}{their} theoretical performance remains largely unexamined.
\textcolor{black}{Consequently, it is difficult to evaluate their performance when ground-truth radio maps are unavailable, as is often the case in practice}.
To bridge this gap, we first formulate radio map estimation as a non-linear matrix completion problem.
We then \textcolor{black}{derive a theoretical expression for the estimation error of a diffusion model, capturing the effects of the mismatch between the training and deployment environments, the model design, and the sampling strategy on the quality of the estimated radio map.}
Moreover, considering that the derived error expression depends on certain information that is difficult to obtain in practice, we propose an empirical approximation that is readily computable from observable data.
Finally, extensive simulations demonstrate that the empirical formula closely approximates the theoretical error expression, validating its effectiveness for practical deployment.
\textcolor{black}{Our results provide guidelines not only for evaluating the performance of a diffusion model when ground truth is unavailable, but also for determining how many sensors are required to spatially sample the region to achieve a given estimation accuracy. 
They further reveal a critical sensor count beyond which the estimation performance of diffusion models converges.}
\end{abstract}

\begin{IEEEkeywords}
Radio map estimation, spectrum cartography, diffusion model, generative AI.
\end{IEEEkeywords}

\section{Introduction}\label{sec:introduction}
A radio map characterizes key parameters of wireless communication channels, such as Received Signal Strength (RSS), at every location of a geographical area of interest~\cite{RME}.  
It can significantly enhance the performance of various wireless applications, by supporting dynamic spectrum access~\cite{Peng24:TCE:Access}, efficient spectrum sharing~\cite{Matar24:JSAC:Sharing}, fingerprint-based localization~\cite{CSELM-QE}, and intelligent interference management~\cite{Huang24:TWC:inference}.

The radio map estimation problem, specifically the sampling-based radio map estimation problem, requires the construction of a radio map based on the limited number of radio samples collected by sensors sparsely distributed in the area~\cite{RadioUNet,autoencoder,RadioGAT}.
This problem is challenging due to ultra-low sampling rates:
The high cost of capable radio sensing hardware (e.g., even a low-end sensor like the Ettus USRP B200mini can cost over $1, 000$ USD~\cite{wifi-diffusion} ) severely limits the number of sensors that can be deployed in practice. 
Consequently, radio samples are ultra-sparse compared to the high resolution of radio maps to be estimated. 
An effective solution of radio map estimation must be able to achieve accurate estimations with ultra-low sampling rates, making efficient use of extremely limited radio samples.

Traditional radio map estimation approaches are based on physics modeling~\cite{ray-tracing,dominant-path,eg,Gaussian-Process} or interpolation~\cite{rbf,spline,kriging}, which utilize mathematical and statistical methodologies for estimation. 
However, these methods often fail to construct high-quality radio maps for real-world environments that are dynamic and geographically complex. 
This is because they do not account for certain critical prior information, such as the layout of obstacles, which can impact radio signal propagation. 

Recent advances in generative Artificial Intelligence (AI), especially diffusion models, lead to a significant shift in how radio maps can be estimated.
There are increasingly works proposing diffusion-based algorithms for radio map estimation.
For example, Liu \emph{et al}. \cite{wifi-diffusion} designed a diffusion-based generative framework for fine-grained radio map estimation. 
Wang \emph{et al}. \cite{RadioDiff} proposed a decoupled diffusion model with an adaptive fast Fourier transform module for radio map estimation.  
Luo \emph{et al}. \cite{RM-Gen} introduced a conditional denoising diffusion probabilistic model for radio map estimation from sparse measurements and transmitter locations.
Qiu \emph{et al}. \cite{IRDM} developed a diffusion model for interpolating indoor path loss radio maps from a limited number of reference points.

Despite the growing number of research works in this area, almost all of them share a common limitation: they designed diffusion-based models and evaluated them \textcolor{black}{by comparing the estimated radio maps with ground-truth radio maps or with a relatively large number of RSS measurements.
However, in scenarios where ground-truth radio maps or a large number of RSS measurements are unavailable, which is often the case in practice due to the high cost of sensors, it is difficult to evaluate the performance of the designed diffusion models.
To bridge this gap, in this paper, we present a theoretical performance analysis of diffusion models for radio map estimation}.
To the best of our knowledge, the sole work addressing the theoretical performance of radio map estimation is by Romero \emph{et al}.~\cite{RME-theory}.
Their study first provides a theoretical analysis of the spatial variability of radio maps and then derives upper bounds for the estimation error of certain interpolation-based methods.
However, the results in~\cite{RME-theory} are limited to free-space propagation environments and are specific to interpolation-based techniques.
As such, they are not applicable to practical, complex environments where factors such as obstacle layouts critically impact propagation.
They cannot be extended to diffusion-based algorithms, either.

Specifically, we present a theoretical foundation for radio map estimation by diffusion models, based on information theory and data-driven methodologies.
The main contributions of this paper are summarized as follows: 

\begin{itemize}
\item We formulate radio map estimation using diffusion models as a non-linear matrix completion problem.
Based on this formulation, we derive a \textcolor{black}{theoretical expression for the minimum estimation error} achievable by diffusion models, which is governed by the discrepancy between the deployment distribution and the true underlying radio propagation law.
We then extend this expression to incorporate the effects of model design and sampling strategy, yielding a \textcolor{black}{generalized expression for the estimation error} that jointly captures both distributional mismatch, model limitations, and measurement sparsity.

\item Considering that the derived theoretical error \textcolor{black}{expressions} rely on certain information that is difficult to obtain in practice, we derive empirical formulas that approximate \textcolor{black}{these expressions} and are readily computable from observable data.
Extensive simulations demonstrate that these empirical formulas closely approximate the theoretical error \textcolor{black}{expressions}, validating their utility for practical deployment.

\item \textcolor{black}{Our results provide guidelines not only for evaluating the estimation performance of a diffusion model when ground-truth radio maps are unavailable and only a limited number of sensors can be deployed to collect radio samples, but also for determining the number of sensors required for spatial sampling of the region to achieve a given estimation accuracy. They further reveal a critical number of sensors beyond which the estimation performance of diffusion models converges.}
\end{itemize}

The rest of the paper is organized as follows.
In Section~\ref{sec:related}, we introduce related works.
In Section~\ref{sec:rem}, we illustrate the problem definition of radio map estimation.
In Section~\ref{sec:Preliminaries}, we present the necessary background on matrix completion.
In Section~\ref{sec:rem-matrix}, we formulate the diffusion-based radio map estimation problem as a non-linear matrix completion problem.
In Section~\ref{sec:con}, we provide a convergence analysis of diffusion models for radio map estimation.
In Section~\ref{sec:bounds}, we derive theoretical \textcolor{black}{expressions} for radio map estimation errors using diffusion models.
In Section~\ref{sec:Empirical-formula}, we propose empirical formulas to approximate our theoretical error \textcolor{black}{expressions}.
In Section~\ref{sec:evaluation}, we evaluate the proposed empirical approximation formulas via extensive simulations.
Finally, Section~\ref{sec:conclusion} concludes the paper.

\section{Related Work}\label{sec:related}

The problem of radio map estimation can generally be classified into two categories: sampling-free problem and sampling-based problem. 
The sampling-free problem uses the information from the transmitters to estimate radio maps. 
Existing studies solving these problems include ~\cite{ray-tracing,dominant-path,eg}.
The sampling-free problem assumes knowledge of the information (e.g., number, locations, and transmitting power) of transmitters.
This assumption may not hold in practice, where there often exist many transmitters unknown to users.
In contrast, the sampling-based problem uses a limited number of radio samples collected by sensors sparsely distributed within the region of interest, as well as certain prior information like the layout of obstacles, to estimate radio maps.
In this paper, we focus on the sampling-based radio map estimation problem.
This problem is challenging because of ultra-low sampling rates: The high cost of capable radio sensing hardware severely limits the number of sensors that can be deployed in practice.
Consequently, radio samples are ultra-sparse compared to the high resolution of radio maps to be estimated.

Traditional radio map estimation approaches are based on physics modeling or interpolation, which utilize mathematical and statistical methodologies for estimation.
Well-known methods include Ordinary Kriging (OK)~\cite{kriging}, Radial Basis functions (RBF)~\cite{rbf}, and
Nearest Neighbor (NN)~\cite{KNN}:
\begin{itemize}
    \item OK: It uses semivariogram functions to describe the spatial autocorrelation and estimates the RSS values of unknown points by weighted average of the observed data of known points.

    \item RBF: By creating a network of basis functions around the known data points, it uses a linear combination of these basis functions to approximate the RSS values of unknown points.

    \item NN: It utilizes a limited set of known radio samples to estimate RSS values at unknown locations, relying on the nearest neighbor principle.
\end{itemize}

Other interpolation-based approaches include Spline~\cite{spline} and Inverse Distance Weighting (IDW)~\cite{IDW}.

\textcolor{black}{Interpolation-based methods are typically based on explicit mathematical rules, with their theoretical foundations rooted in spatial statistics and deterministic interpolation principles. 
They offer the advantages of being concise in form and easy to deploy. 
However, they struggle to effectively model and capture the influence of complex factors, such as obstacle layouts, on radio propagation. 
Consequently, the radio maps they estimate are often of low quality in complex real-world environments.} 

Recent advances in deep learning~\cite{RadioUNet, autoencoder, IQ-Data}, particularly generative AI, have led to a significant shift in how radio maps can be estimated.
There are increasingly works proposing generative AI-based algorithms for radio map estimation, surpassing traditional physics modeling-based or interpolation-based approaches. 
Well-known methods include CollaboRadio~\cite{CollaboRadio} and RME-GAN~\cite{RME-GAN}.
\begin{itemize}
    \item CollaboRadio: It is a device-edge-cloud collaboration framework, with a UNet-based small model and a Transformer-based large model proposed.

    \item RME-GAN: It is a two-phase conditional Generative Adversarial Network (GAN) framework that integrates model-based interpolation and learning-based refinement for radio map estimation.
\end{itemize}

Other generative AI-based approaches include RadioMamba \cite{RadioMamba}, TVR \cite{TVR}, UniRM \cite{UniRM}, RadioLAM~\cite{RadioLAM}, and RadioProphet~\cite{RadioProphet}.

\textcolor{black}{Generative AI methods adopt a data-driven paradigm: they learn intricate distributional patterns of radio propagation from large amounts of data and construct an implicit model of the underlying propagation mechanisms. 
This paradigm differs fundamentally from interpolation-based approaches: the latter rely on predefined, simplified mathematical assumptions, whereas generative AI methods infer complex mapping relationships from data. 
Consequently, the two paradigms have markedly different theoretical foundations and applicable regularities.}

\textcolor{black}{Well-known generative AI approaches include GANs, Variational AutoEncoders (VAEs), and diffusion models. 
For radio map estimation, GANs can produce high-quality results, but their training process is prone to mode collapse and lacks interpretable convergence guarantees. 
VAEs introduce bias due to the approximation of the posterior distribution, making the generation error difficult to characterize accurately. 
Compared with GANs and VAEs, diffusion models are particularly prominent. 
They are built upon a unified framework of Stochastic Differential Equations (SDEs) or Ordinary Differential Equations (ODEs). 
This structural property enables diffusion models to achieve superior estimation performance compared with other generative models.}
Diffusion models have been successfully applied to radio map estimation. 
Well-known approaches include WiFi-Diffusion~\cite{wifi-diffusion}, RadioDiff~\cite{RadioDiff}, RM-Gen~\cite{RM-Gen}, and IRDM~\cite{IRDM}.

\begin{itemize}
    \item WiFi-Diffusion: It first uses a diffusion model to generate a candidate set of radio maps with diverse qualities, and then employs principles of radio propagation to identify the most accurate map.

    \item RadioDiff: It is a decoupled diffusion model with an adaptive fast Fourier transform module used for radio map estimation.

    \item RM-Gen: It is a conditional denoising diffusion probabilistic model that generates complete radio maps using sparse RSS fragments and transmitter locations, incorporating an environment-aware method for selecting critical data pieces.
    
    \item IRDM: It is a generative diffusion model designed for indoor radio map interpolation that recovers complete path loss maps from a limited number of reference points, using a multi-stage training strategy and online data augmentation.
\end{itemize}

For these diffusion-based radio map estimation studies, they all adhere to a common profile: they designed diffusion-based models and evaluated them \textcolor{black}{by comparing the estimated radio maps against ground-truth radio maps or against a relatively large number of collected radio samples. 
However, in scenarios where ground-truth radio maps or a large number of radio samples are unavailable, which is often the case in practice due to the high cost of sensors, it is difficult to evaluate the performance of the designed models. 
To bridge this gap, we present a theoretical performance analysis of diffusion models that quantifies the estimation error without requiring access to ground-truth radio maps or a large number of radio samples}.

To the best of our knowledge, the sole work addressing the theoretical performance of radio map estimation is by Romero \emph{et al}.~\cite{RME-theory}.
Their study first provides a theoretical analysis of the spatial variability of radio maps and then derives upper bounds for the estimation error of certain interpolation-based methods.
However, the results in~\cite{RME-theory} are limited to free-space propagation environments and are specific to interpolation-based techniques.
As such, they are not applicable to practical, complex environments where factors such as obstacle layouts critically impact propagation.
They cannot be extended to diffusion models, either.

\section{Problem Definition: Radio Map Estimation}\label{sec:rem}

\begin{figure}[]
	\centering
	\includegraphics[width=\linewidth]{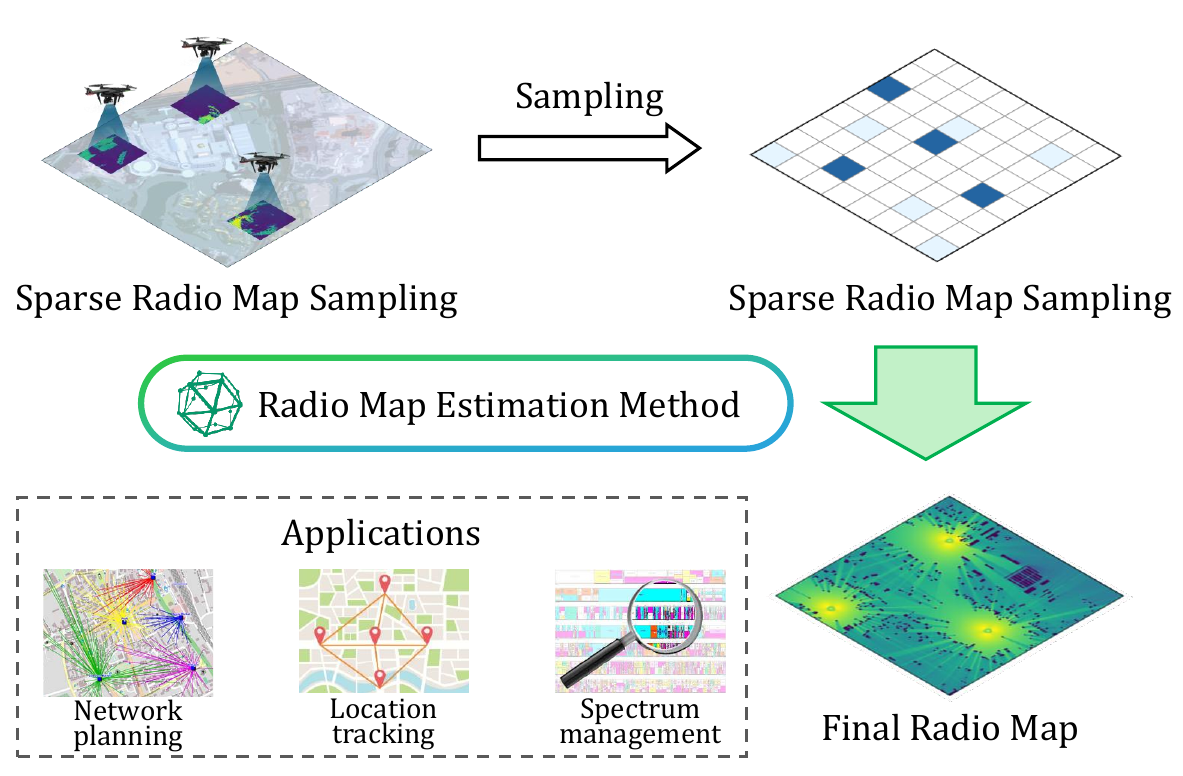}
	\caption{An illustration of the radio map estimation problem.}
	\label{fig:REM}
\end{figure}

Consider a two-dimensional (2D) region of interest $\mathcal{X}\subset\mathbb{R}^2$. 
A set $\mathcal{S}=\{1,2,\cdots,S\}$ of $S$ sensors is deployed over $\mathcal{X}$, where each sensor $s\in\mathcal{S}$ is located at a distinct coordinate $x_s\in\mathcal{X}$.
Each sensor $s$ measures RSS $r_s$ in a specific frequency band at its location $x_s$ (it is assumed that fast-fading effects are averaged out through minor temporal or spatial variations during measurement).
These radio samples (RSS measurements) $\{r_s,\forall s\in\mathcal{S}\}$ are reported to a fusion center, which may be, e.g., a base station or a cloud server, depending on the application.

Consider that the continuous space $\mathcal{X}$ is discretized into an $m\times n$ grid.
A radio map can be represented by an $m\times n$ matrix $\mathbf{M}$, where its element at index $(h,w)$ ($h\in\{0,1,\cdots,m-1\}$ and $w\in\{0,1,\cdots,n-1\}$) corresponds to the RSS at the respective spatial coordinate.
Let $\Psi$ denote an $m\times n$ binary mask, where its element at index $(h,w)$ is 1 if a sensor is located at $(h,w)$; otherwise, it is $0$.
Denote $\odot$ as the Hadamard product.
Then, the partial observation matrix $\mathbf{M}\odot\Psi$ yields an $m\times n$ matrix whose non-zero entries correspond to the RSS values measured by the sensors at their respective locations.  

The problem of radio map estimation requires the fusion center to use the samples collected by sensors sparsely distributed within $\mathcal{X}$, i.e., use the set $\{r_s,\forall s\in\mathcal{S}\}$, to estimate the radio map $\mathbf{M}$.
Specifically, as depicted in Fig.~\ref{fig:REM}, radio map estimation requires the fusion center to use the incomplete observations $\mathbf{M}\odot\Psi$ to estimate $\mathbf{M}$ which presents the RSS at every location within the region $\mathcal{X}$.

\section{Preliminaries: Matrix Completion}\label{sec:Preliminaries}
Our analysis of using diffusion models to solve the radio map estimation problem is based on a matrix completion-based formulation.
In order to describe the formulation, it is necessary to first introduce matrix completion in this section.

The problem of matrix completion requires to recover missing entries in partially observed matrices~\cite{MC}. 
It has been widely applied to numerous tasks, such as collaborative filtering~\cite{collaborative}, image recovery and inpainting~\cite{inpainting}, and classification~\cite{classification}. 

\subsection{Low-Rank Matrix Completion}
A fundamental assumption underlying traditional matrix completion methods is that the target data matrix is low-rank. This low-rank property allows missing entries to be recovered by minimizing the matrix rank.

\begin{figure}[]
	\centering
	\includegraphics[width=1\linewidth]{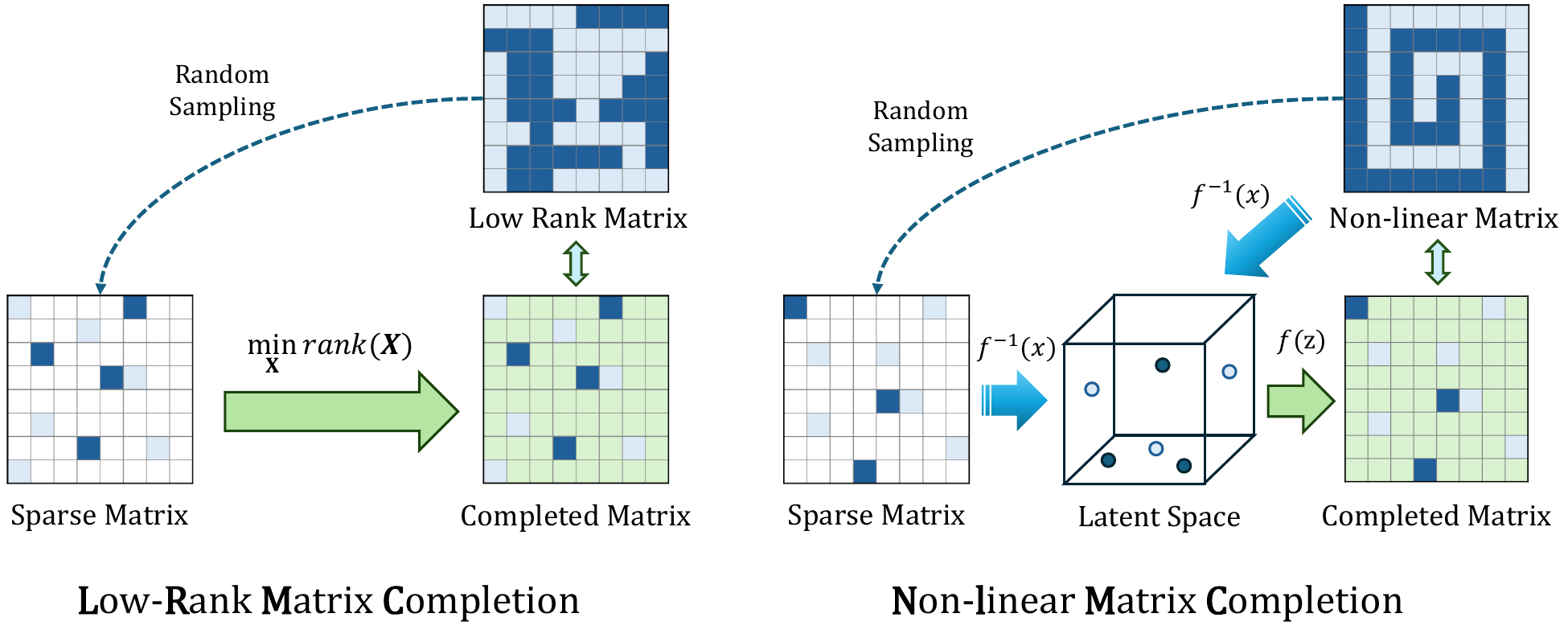}
    \caption{An illustration of the matrix completion problem.}
	\label{fig:MC}
\end{figure}

Specifically, theoretical studies have shown that, under certain constraints on the missing rate, matrix rank, and sampling scheme, missing entries in a low-rank matrix can be exactly recovered by addressing a rank minimization problem~\cite{LRMC}. 
Mathematically, for a matrix $\mathbf{M}_{m\times n}$ and an observation position set $\Omega$, define a sparse observation mask as  
$$[\Psi]_{ij} =  
\begin{cases}  
1 & \text{if } (i,j) \in \Omega \\  
0 & \text{otherwise}.  
\end{cases}$$  
Then, for matrix $\mathbf{X}$ to be recovered, the low-rank matrix completion problem can be formulated as:  
\begin{equation}
    \begin{aligned}  
&\min_{\mathbf{X}} \text{rank}(\mathbf{X}) \\  
&\text{subject to } \mathbf{X} \odot \Psi = \mathbf{M} \odot \Psi.  
\end{aligned} \nonumber 
\end{equation}
where $\mathbf{X}_{m\times n}$ is the decision variable.

\subsection{Non-Linear Matrix Completion}
Traditional matrix completion methods often fail to recover missing entries in partially observed matrices by minimizing the matrix rank when the matrix has a high rank.
To address this and recover missing elements in high-rank matrices, non-linear matrix completion jointly optimizes both a non-linear mapping from the high-rank matrix to a low-dimensional latent space and the latent structure~\cite{NLMC}.
They assume that the data of the high-rank matrix is given by a non-linear latent variable model, and have proved that minimizing the latent space dimension can exactly solve the matrix completion problem.

Mathematically, let $f(\cdot)$ be a non-linear transformation that relates a low-dimensional latent space and the high-rank matrix. 
The non-linear matrix completion problem is:  

\begin{equation}
    \begin{aligned}  
&\min_{\mathbf{X}} \text{dim}(\mathcal{Z}) \\  
&\text{subject to } \mathbf{X}=f(\mathcal{Z}),\ \mathbf{X}\odot \Psi  = \mathbf{M} \odot \Psi.  
\end{aligned}\nonumber
\end{equation}
where $\mathcal{Z}$ is the latent space variable.

\textcolor{black}{For the above problem, let the sampling rate be $\rho=\frac{|\Omega|}{mn}$. A necessary condition for the problem to be solvable is~\cite{NLMC-DMF}:
\begin{equation}
    \text{rank}(\mathcal{Z})\leq \rho\cdot \min(m,n).\label{eq:rank}
\end{equation}}

Depending on the non-linear mapping model, existing approaches for non-linear matrix completion can be basically categorized into two categories, i.e., kernel trick approaches~\cite{NLMC} and Deep Matrix Factorization (DMF) approaches~\cite{NLMC-DMF}.

\smallskip
\noindent \textbf{Kernel Trick-Based Non-Linear Matrix Completion} \quad
It maps the sparsely sampled matrix into a high-dimensional space using kernel methods, thereby constructing a high-dimensional yet low-rank matrix. Matrix completion is then performed in this high-dimensional space via low-rank matrix completion. 

Specifically, to handle non-linearity, a non-linear function \( \phi(\cdot) \) is defined to map \( X \) into a feature space \( \mathcal{F} \):
$\phi(\cdot) : \mathbb{R}^m \to \mathcal{F}^l,$
where \( \mathcal{F} \) is an inner product space whose dimension \( l \) can be arbitrarily large or even infinite. Data in \( \mathcal{F} \) can then be processed using linear techniques. The matrix formed by the data in the feature space is
\[
W = \Phi(X) = [\phi(x_1), \phi(x_2), \cdots, \phi(x_n)],
\]
which is expected to be low-rank or approximately low-rank. This implies that in the high-dimensional feature space, \( W \) originates from a low-dimensional subspace and can be represented by a linear transformation of a small number of latent variables. This aligns with the situation in the input space, where \( X \) can be expressed via a non-linear transformation \( f(\cdot) \) of a small set of latent variables \( \mathcal{Z} \). 
Therefore, minimizing the rank of \( W \) is equivalent to minimizing the dimension of non-linear latent variables of \( X \).

As a result, the missing entries of \( X \) can be recovered by solving the following problem:

\begin{equation}
    \min_{X} \text{rank}(W), \ \text{s.t.} \ W = \Phi(X),\ \mathbf{X}\odot \Psi  = \mathbf{M} \odot \Psi, \ (i, j) \in \Omega,\nonumber
\end{equation}

This problem is NP-hard but can be approximated as:

\begin{equation}
    \min_{X} \|W\|_{S_p}^p, \  \text{s.t.} \  W = \Phi(X), \  \mathbf{X}\odot \Psi  = \mathbf{M} \odot \Psi, \  (i, j) \in \Omega.\nonumber
\end{equation}

Here, the Schatten \( p \)-norm serves as a non-convex relaxation of the rank, and is defined as:

\begin{equation}
    \|W\|_{S_p} = \left( \sum_{i=1}^{\min(n,l)} (\sigma_i(W))^p \right)^{1/p}.\nonumber
\end{equation}

Kernel trick-based non-linear matrix completion approaches are limited by their computational inefficiency and moderate reconstruction performance.

\smallskip
\noindent \textbf{DMF-Based Non-Linear Matrix Completion} \quad
It builds upon the non-linear latent variable model by constructing a deep-structured neural network. 
It takes low-dimensional unknown latent variables as input and partially observed variables as output. 
By simultaneously optimizing the latent variables and network parameters to minimize the reconstruction error of the observed entries, it recovers the missing ones.
Compared to kernel trick-based approaches, DMF-based approaches have significantly lower computational complexities and are suitable for the completion of large-scale matrices. 
The formal expression of DMF-based non-linear matrix completion is:
\[
\begin{split}
&\min_{\mathcal{Z},f} \pi(f) + \frac{\beta}{2n}\|\mathcal{Z}\|_{F}^{2}, \\
&\text{s.t. } [f(\mathcal{Z})]_{i,j} = X_{i,j},\ (i,j) \in \Omega,
\end{split}
\]
where \(\pi(f)\) denotes the penalty or constraint on \(f(\cdot)\), and \(\beta\) is the regularization parameter for \(\mathcal{Z}\). Regularization of both \(f(\cdot)\) and \(\mathcal{Z}\) is necessary since there are infinitely many solutions of \(f(\cdot)\) and \(\mathcal{Z}\) that can reproduce the observed entries of \(X\). 

In practice, the observed entries of \(X\) are often contaminated by small noise, and it is impractical for \([f(\mathcal{Z})]_{i,j}\) to exactly equal \(X_{i,j}\). It suffices to minimize the discrepancy between \([f(\mathcal{Z})]_{i,j}\) and \(X_{i,j}\):
\[
\min_{\mathcal{Z},f} \frac{1}{2n} \|\Psi \odot (X - f(\mathcal{Z}))\|_{\mathrm{F}}^{2} + \frac{\beta}{2n} \|\mathcal{Z}\|_{\mathrm{F}}^{2} + \lambda \pi(f).
\]
where \(\lambda\) is a regularization parameter. 
The non-linear function can be approximated using a neural network:
\[
\begin{split}
&\min_{\mathcal{Z},\Theta} \frac{1}{2n} \sum_{i=1}^{n} \| \Psi_{i} \odot (x_{i} - g^{(K+1)}(g^{(K)}(\cdots g^{(1)}(z_{i},\Theta^{(1)}) \\
&\qquad \cdots, \Theta^{(K)}), \Theta^{(K+1)})) \|^{2} + \frac{\beta}{2n} \|\mathcal{Z}\|_{\mathrm{F}}^{2} + \frac{\lambda}{2} \sum_{j=1}^{K+1} \|W^{(j)}\|_{\mathrm{F}}^{2},
\end{split}
\]
where \(\Theta^{(j)} = (W^{(j)}, b^{(j)})\), \(g^{(j)}(t,\Theta^{(j)}) = g^{(j)}(W^{(j)}t + b^{(j)})\), for \(j=1, 2, \cdots, K+1\), and \(K\) is the number of hidden layers.

Although DMF-based methods offer theoretical advantages (e.g., guaranteed complete recovery above a certain sampling threshold), they remain limited when sampling falls below that threshold, where insufficient information prevents exact reconstruction.
Furthermore, DMF differs fundamentally from typical deep learning models like CNNs used for image classification. CNNs are trained on massive image datasets to learn general feature extraction capabilities applicable to all images. In contrast, DMF aims to precisely capture the internal structure and patterns of a specific matrix. It is tailored only to that particular matrix and is not designed to generalize to other unrelated matrices.

\section{Matrix Completion-Based Formulation for Diffusion-Based Radio Map Estimation}\label{sec:rem-matrix}
Although standard non-linear matrix completion can take sparse sensor measurements into consideration, they fail to account for certain prior information that is critical for radio map estimation, e.g., well-established radio propagation principles.
To address this limitation, we are motivated by Inductive Matrix Completion (IMC)~\cite{IMC}, which utilizes side information to construct non-linear mappings. In this section, we thereby introduce a non-linear matrix completion formulation for the problem of using diffusion models to estimate radio maps.

\begin{figure}[]
	\centering
	\includegraphics[width=1\linewidth]{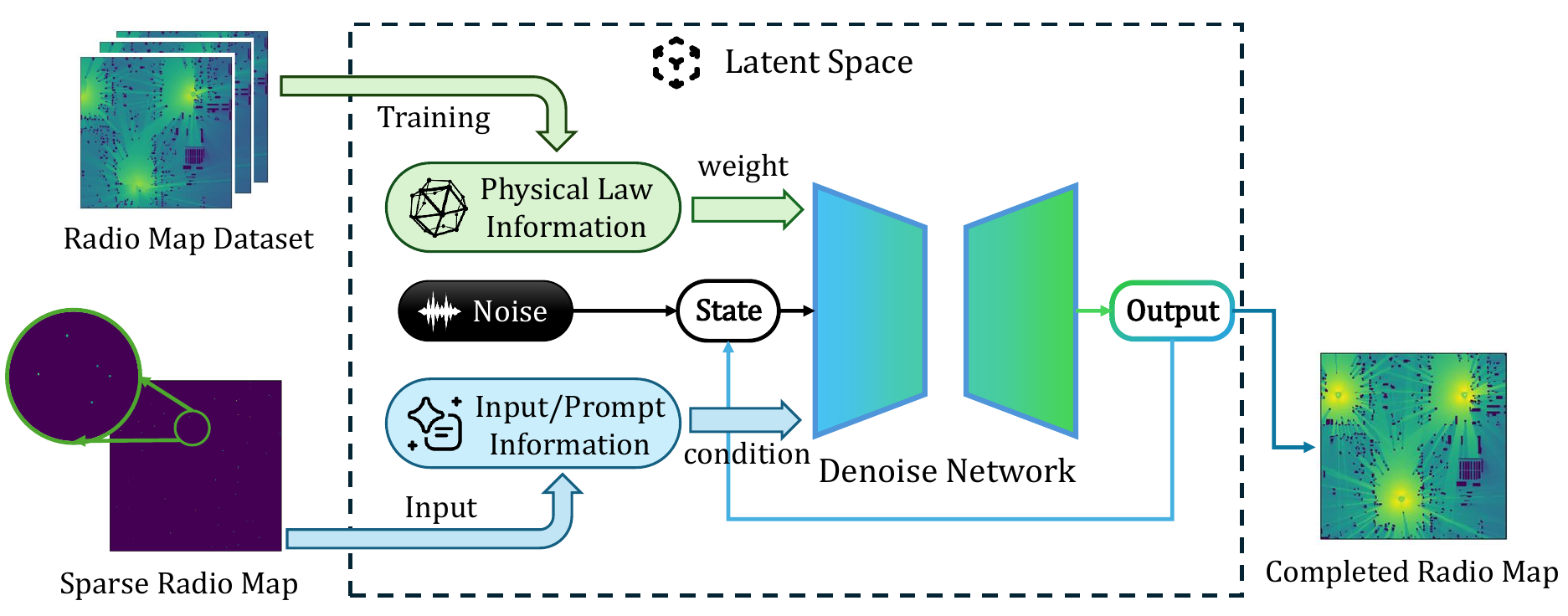}
    \caption{An illustration of using diffusion models for radio map estimation.}
	\label{fig:GAIMC}
\end{figure}

A radio map is in fact a high-rank matrix of RSS values, derived from a latent and potentially unknown propagation law conditioned on environmental factors, such as building layouts and terrain, which influence signal behavior.
As shown in Fig.~\ref{fig:GAIMC}, when a diffusion model is applied to radio map estimation, it learns the complex mapping from radio propagation law to RSS matrices during training, using a dataset that encodes the propagation law alongside ground-truth RSS matrices. 
This learning process is aligned intuitively with the problem of non-linear matrix completion.

Mathematically, let $\mathbf{M}_{m\times n}$ denote the target radio map to be estimated, and let $G_\theta(\cdot)$ represent a diffusion model parameterized by $\theta$. Given a binary observation mask $\Psi$ indicating sensor locations.
Let $\mathcal P$ be the underlying distribution of radio maps, governed by a latent and potentially unknown radio propagation law. \textcolor{black}{Then we have: 
\begin{gather*}\forall \, M_1, M_2 \in \operatorname{supp}(\mathcal{P}), \quad M_1 \odot \Psi = M_2 \odot \Psi \\\implies \|M_1 - M_2\|^2 \le \varepsilon.\end{gather*}}
Let $\mathbf{N_Z}$ denote the input noise to the diffusion model. The problem of radio map estimation by diffusion models can be formulated as:

\begin{equation}
    \begin{aligned}&\min_{\mathbf{\theta, \mathbf{X}}}\mathbf{W}(G_\theta, \mathcal{P})  ,  \\&\text{subject to } \mathbf{X}=G_\theta(\mathbf{N_Z};\mathbf{M} \odot \Psi),\ \mathbf{X}\odot \Psi  = \mathbf{M} \odot \Psi.\end{aligned}\label{equ:GAIMCytarget}
\end{equation}
where \textcolor{black}{$\mathbf{W}$ denotes the Wasserstein distance}, which measures the similarity between different distributions. 
Specifically, for two probability measures $\mu, \nu \in P(\mathbb{R}^d)$,  
\begin{equation}
\mathbf{W}(\mu, \nu) = \inf_{\gamma \in \Gamma(\mu, \nu)} \int_{\mathbb{R}^d \times \mathbb{R}^d} \|x - y\| \, d\gamma(x, y)
\end{equation}
where $\Gamma(\mu, \nu)$ represents the set of all coupling measures with marginals $\mu$ and $\nu$.

\section{Convergence Analysis of Diffusion Models for Radio Map Estimation}\label{sec:con}

Diffusion models address the radio map estimation problem by gradually injecting noise to smoothly transform the complex data distribution of radio propagation laws into a known prior distribution, and then reversing the process by gradually removing noise to reconstruct the data distribution from the prior.
This idea aligns with the framework of SDEs.
In this section, we establish the convergence of the diffusion model for radio map estimation using the SDE optimization theory, providing a theoretical foundation for the analysis of the estimation errors of diffusion models in subsequent sections.

The objective of the diffusion model is to construct a diffusion process $\{ \mathbf{x}(t) \}_{t=0}^T$, indexed by a continuous time variable $t \in [0, T]$, such that $\mathbf{x}(0) \sim p_0 = \mathcal{P}$ (the data distribution of radio propagation) and $\mathbf{x}(T) \sim p_T$ (a tractable prior distribution with a known analytic form). This diffusion process can be modeled as the solution to an It\^{o} SDE:
\begin{equation}
\mathrm{d}\mathbf{x} = \mathbf{f}(\mathbf{x}, t)\,\mathrm{d}t + g(t)\,\mathrm{d}\mathbf{w},
\end{equation}
where $\mathbf{w}$ denotes the standard Wiener process (also known as Brownian motion), $\mathbf{f}(\cdot, t): \mathbb{R}^d \to \mathbb{R}^d$ is a vector-valued function referred to as the drift coefficient of $\mathbf{x}(t)$, and $g(\cdot): \mathbb{R} \to \mathbb{R}$ is a scalar function known as the diffusion coefficient of $\mathbf{x}(t)$. 
As long as these coefficients are globally Lipschitz continuous in both state and time, the SDE admits a unique strong solution~\cite{Oksendal2003}. 

Let $p_t(\mathbf{x})$ be the probability density of $\mathbf{x}(t)$, and $p_{st}(\mathbf{x}(t)\,|\,\mathbf{x}(s))$ be the transition kernel from $\mathbf{x}(s)$ to $\mathbf{x}(t)$ for $0 \le s < t \le T$. Typically, $p_T$ is an unstructured prior that contains no information about $p_0$, such as a Gaussian distribution with fixed mean and variance.
By initializing from samples $\mathbf{x}(T) \sim p_T$ and reversing the diffusion process, we can obtain samples $\mathbf{x}(0) \sim p_0$. A seminal result by Anderson~\cite{Anderson1982} establishes that the reverse of a diffusion process is itself a diffusion process evolving backward in time, governed by the following reverse-time SDE:
\begin{equation}
\mathrm{d}\mathbf{x} = \left[ \mathbf{f}(\mathbf{x}, t) - g(t)^2 \nabla_{\mathbf{x}} \log p_t(\mathbf{x}) \right] \mathrm{d}t + g(t)\,\mathrm{d}\bar{\mathbf{w}},
\label{equ:rSDE}
\end{equation}
where $\bar{\mathbf{w}}$ is a standard Wiener process when time flows backwards from $T$ to $0$, and $\mathrm{d}t$ represents an infinitesimal negative time step. Once the score $\nabla_{\mathbf{x}} \log p_t(\mathbf{x})$ of each marginal distribution is known for all $t$, the reverse diffusion process can be derived from Eq.~\eqref{equ:rSDE} and simulated to draw samples from $p_0$.

The score of a distribution can be estimated by training a score-based model on data samples using score matching~\cite{Song2019}. To estimate $\nabla_{\mathbf{x}} \log p_t(\mathbf{x})$, one can train a time-dependent score-based model $\mathbf{s}_{\boldsymbol{\theta}}(\mathbf{x}, t)$ via a continuous generalization of noise
conditional score network (NCSN) training target~\cite{Song2019}, and evidence lower bound (ELBO) of denoising diffusion probabilistic models (DDPM)~\cite{DDPM}:
\begin{equation}
\begin{aligned}
    \boldsymbol{\theta}^* &= \arg\min_{\boldsymbol{\theta}} \mathbb{E}_{t} \{ \lambda(t) \mathbb{E}_{\mathbf{x}(0)} \mathbb{E}_{\mathbf{x}(t) \mid \mathbf{x}(0)} [ \| \mathbf{s}_{\boldsymbol{\theta}}(\mathbf{x}(t), t)  \\& - \nabla_{\mathbf{x}(t)} \log p_{0t}(\mathbf{x}(t) \mid \mathbf{x}(0)) \|_2^2 ] \}.
\end{aligned}
\label{equ:SDEtarget}
\end{equation}
Here, $\lambda: [0, T] \rightarrow \mathbb{R}_{>0}$ is a positive weighting function, $t$ is uniformly sampled over $[0, T]$, $\mathbf{x}(0) \sim p_0(\mathbf{x})$, and $\mathbf{x}(t) \sim p_{0t}(\mathbf{x}(t) \mid \mathbf{x}(0))$. Given sufficient data and model capacity, score matching guarantees that the optimal solution to Eq.~\eqref{equ:SDEtarget}, denoted as $\mathbf{s}_{\boldsymbol{\theta}^*}(\mathbf{x}, t)$, equals to $\nabla_{\mathbf{x}} \log p_t(\mathbf{x})$ for almost all $\mathbf{x}$ and $t$.

In practice, knowledge of the transition kernel $p_{0t}(\mathbf{x}(t) \mid \mathbf{x}(0))$ is typically required to efficiently optimize Eq.~\eqref{equ:SDEtarget}. 
When $\mathbf{f}(\cdot, t)$ is affine, this kernel is Gaussian with closed-form expressions for its mean and variance, obtainable via standard techniques. For instance, the perturbation kernels $\{p_{\alpha_i}(\mathbf{x} \mid \mathbf{x}_0)\}_{i=1}^{N}$ used in DDPM, one well-known diffusion model, correspond to the discrete Markov chain:
\begin{equation}
\mathbf{x}_i = \sqrt{1 - \beta_i}\, \mathbf{x}_{i-1} + \sqrt{\beta_i}\, \mathbf{z}_{i-1}, \quad i = 1, \dots, N.
\end{equation}
In the limit as $N \rightarrow \infty$, this discrete process converges to the following SDE:
\begin{equation}
\mathrm{d}\mathbf{x} = -\frac{1}{2} \beta(t)\, \mathbf{x}\, \mathrm{d}t + \sqrt{\beta(t)}\, \mathrm{d}\mathbf{w}.
\end{equation}

To ensure that the above optimization process in Eq.~\eqref{equ:SDEtarget} effectively solves the problem of radio map estimation formulated in Eq.~\eqref{equ:GAIMCytarget}, it is essential to establish the convergence of the diffusion model.
Bortoli~\cite{convergence} \textcolor{black}{has} proved that if the following $4$ assumptions hold:

\begin{itemize}
    \item \textbf{A1.} The data distribution $\mathcal{P}$ is supported on a compact set $\mathcal{M}$, and $0 \in \mathcal{M}$,

    \item \textbf{A2.} $t \mapsto \beta_{t}$ is continuous, non-decreasing, and there exists $\bar{\beta} > 0$ such that for any $t \in [0,T]$, $1/\bar{\beta} \leq \beta_{t} \leq \bar{\beta}$,

    \item \textbf{A3.} There exist $\boldsymbol{s} \in \mathrm{C}([0,T] \times \mathbb{R}^{d}, \mathbb{R}^{d})$ and $\mathtt{M} \geq 0$ such that for any $t \in [0,T]$ and $x_{t} \in \mathbb{R}^{d}$,
\begin{equation}
\|\boldsymbol{s}(t, x_{t}) - \nabla \log p_{t}(x_{t})\| \leq \mathtt{M}(1 + \|x_{t}\|)/\sigma_{t}^{2},
\end{equation}

    \item \textbf{A4.} $\gamma_{k} \sup_{v \in [T-t_{k+1}, T-t_{k}]} \beta_{v}/\sigma_{v}^{2} \leq \delta \leq 1/2$, for any $k \in \{0,\ldots,K-1\}$,
\end{itemize}

\noindent and $T \geq 2\bar{\beta}(1 + \log(1 + \mathrm{diam}(\mathcal{M})))$, $\gamma_{K} = \varepsilon$, and $\varepsilon, \mathtt{M}, \delta \leq 1/32$, there exists $\mathtt{D}_{0} \geq 0$ such that
\begin{equation}
\begin{aligned}
    \mathbf{W}(G_\theta, \mathcal{P}) \leq& \mathtt{D}_{0}(\exp[\kappa/\varepsilon](\mathtt{M} + \delta^{1/2})/\varepsilon^{2} \\&+ \exp[\kappa/\varepsilon]\exp[-T/\bar{\beta}] + \varepsilon^{1/2}),
\end{aligned}
\end{equation}
where $\kappa = \mathrm{diam}(\mathcal{M})^{2}(1 + \bar{\beta})/2$ and
\begin{equation}
\mathtt{D}_{0} = D(1 + \bar{\beta})^{7}(1 + d + \mathrm{diam}(\mathcal{M})^{4})(1 + \log(1 + \mathrm{diam}(\mathcal{M}))),
\end{equation}
where $D$ is a constant, $\mathrm{diam}(\mathcal{M})$ denotes the diameter of the manifold, defined as $\mathrm{diam}(\mathcal{M}) = \sup\{\|x-y\| : x, y \in \mathcal{M}\}$, $\sigma_t$ is the noise scale of the forward SDE in the diffusion process, and $\varepsilon>0$ is a small hyperparameter.

For radio map estimation by diffusion models, \textbf{A1} holds because radio maps follow underlying radio propagation laws, even when these laws are not explicitly known. \textbf{A2} is satisfied in all scheduling schemes used in practice. The explosive behavior as $t \to 0$ considered in \textbf{A2} is properly accounted for in real-world implementations. \textbf{A3} and \textbf{A4} hold for the SDE. 
Consequently, the diffusion model converges when applied to the problem of radio map estimation.

\section{Theoretical Radio Map Estimation Errors of Diffusion Models}\label{sec:bounds}

In this section, we analyze the theoretical performance of diffusion models for radio map estimation. First, recognizing that any estimated radio map will inevitably deviate from the ground truth due to discrepancies between the deployment domain distribution and the true underlying radio propagation law, we derive \textcolor{black}{an expression for the} minimum estimation error achievable by a diffusion model.
Next, we extend this \textcolor{black}{error expression} to take the model design and sampling strategy into consideration, thereby capturing the additional error introduced by model limitations and ultra-sparse sampling.
Finally, we \textcolor{black}{derive conditions on radio sampling for diffusion models to achieve a given accuracy of radio map estimation}.

\subsection{Minimum Estimation Error}

In the context of generative AI, deployment domain distribution refers to the data distribution encountered when a trained model is deployed in a real-world environment. 
For radio map estimation, variations in deployment scenarios can cause the deployment domain distribution to deviate from the true underlying radio propagation law.
Let $\widetilde{\mathcal{P}}$ denote the deployment domain distribution. 
Then, $\mathbf{W}(G_\theta, \widetilde{\mathcal{P}})$ characterizes the minimum error $E_{min}$ of using diffusion models for radio map estimation.
According to the convergence theory of diffusion models, with a suitable estimator, the convergence effect can be related to the recovery effect measured by Mean Squared Error (MSE)~\cite{convergenceMSE}, i.e., we can have:

\begin{equation}\label{eqn:gap}
E_{min}=||\mathbf X_{max}-\widetilde{\mathbf M}||^2,
\end{equation}
where $\mathbf X_{max}$ is the best map (matrix) reconstructed by diffusion models, and $\widetilde{\mathbf M}$ is the map (matrix) corresponding to the deployment domain distribution.

Based on the convergence theory of diffusion models~\cite{SDE, convergence}, which relates the Wasserstein distance to the convergence effect, we derive the following theorem, which \textcolor{black}{presents an expression for} the minimum estimation error achievable by diffusion models.

\begin{theorem}
The minimum estimation error $E_{\text{min}}$ (defined in Eq.~\eqref{eqn:gap}) of using diffusion models for radio map estimation \textcolor{black}{can be approximated as}:
\begin{equation}\label{eqn:EMIN}
E_{\text{min}}~ \textcolor{black}{\approx} ~\mathcal R(\mathcal{P}, \widetilde{\mathcal{P}}) \cdot \mathbf{W}(\mathcal{P}, \widetilde{\mathcal{P}}) + \zeta,
\end{equation}
where $\mathcal R(\mathcal{P}, \widetilde{\mathcal{P}})$ is an influence factor related to the manifold geometry of distributions $\mathcal{P}$ and $\widetilde{\mathcal{P}}$, used to scale the distribution error to the specific data error and defined as:
\begin{equation}
\begin{aligned}
    \mathcal R(\mathcal{P}, \widetilde{\mathcal{P}}) ~=~& C \cdot \left( \frac{\mathrm{diam}(\mathcal{A}) + \mathrm{diam}(\widetilde{\mathcal{A}})}{2} \right) \\&\cdot \left(1 + \log\left(1 + \mathbf{W}(\mathcal{P}, \widetilde{\mathcal{P}})\right)\right)
\end{aligned}
\end{equation}
Here $\mathcal{A}$ and $\widetilde{\mathcal{A}}$ are the supporting manifolds of distributions $\mathcal{P}$ and $\widetilde{\mathcal{P}}$ respectively, $\mathrm{diam}(\cdot)$ denotes the diameter of the manifold, and $C$ is a constant related to the diffusion model parameters. $\zeta$ is an error term reflecting the approximation error between the diffusion model $G_\theta$ and the training distribution $\mathcal{P}$.\label{T:max}
\end{theorem}
\begin{IEEEproof}
    Consider the training distribution $\mathcal{P}$ and the deployment distribution $\widetilde{\mathcal{P}}$, supported on compact manifolds $\mathcal{A}$ and $\widetilde{\mathcal{A}}$ respectively.
According to the triangle inequality property of the Wasserstein distance, we have:
\begin{equation}
|\mathbf{W}(G_\theta, \widetilde{\mathcal{P}})-\mathbf{W}(\mathcal{P}, \widetilde{\mathcal{P}})| \le \mathbf{W}(G_\theta, \mathcal{P}).\nonumber
\end{equation}
As a result, we have:
\begin{equation}
\mathbf{W}(G_\theta, \widetilde{\mathcal{P}})=\mathbf{W}(\mathcal{P}, \widetilde{\mathcal{P}})+\delta_\epsilon\nonumber
\end{equation}
where $\delta_\epsilon \in [-\mathbf{W}(G_\theta, \mathcal{P}),\mathbf{W}(G_\theta, \mathcal{P})]$ is a random variable, whose absolute value decreases with convergence.

Thus, due to the influence of diffusion model convergence, the theoretical minimum error $E_{\min}$ is primarily determined by the inter-distribution distance $\mathbf{W}(\mathcal{P}, \widetilde{\mathcal{P}})$, but needs to be modulated by the manifold geometry to map from the distribution to the specific sampling error. Therefore, we establish:
\begin{equation}
E_{\text{min}} ~\textcolor{black}{\approx} ~\mathcal R(\mathcal{P}, \widetilde{\mathcal{P}}) \cdot \mathbf{W}(\mathcal{P}, \widetilde{\mathcal{P}}) + \zeta\nonumber
\end{equation}
where $\mathcal R(\mathcal{P}, \widetilde{\mathcal{P}})$ is an influence factor related to the manifold geometry of distributions $\mathcal{P}$ and $\widetilde{\mathcal{P}}$, and $\zeta$ represents the contribution of the model approximation error.

According to the convergence theory of diffusion models under the manifold assumption \cite{SDE,convergence}, the Wasserstein distance between the generative distribution $G_\theta$ and the target distribution $\mathcal{P}$ satisfies:
\begin{equation}
\begin{aligned}
    \mathbf{W}(G_\theta, \mathcal{P}) \leq& \mathtt{D}_{0}(\exp[\kappa/\varepsilon](\mathtt{M}+\delta^{1/2})/\varepsilon^{2} \\&+ \exp[\kappa/\varepsilon]\exp[-T/\bar{\beta}] + \varepsilon^{1/2})
\end{aligned}\nonumber
\end{equation}
where the constant $\mathtt{D}_{0}$ depends on the manifold diameter $\mathrm{diam}(\mathcal{A})$, dimension $d$, and diffusion parameters $\bar{\beta}$.

According to the convergence analysis of diffusion models on manifolds, the constant $\mathtt{D}_{0}$ has the form \cite{convergence}:
\begin{equation}
\mathtt{D}_{0} = D(1+\bar{\beta})^{7}(1+d+\mathrm{diam}(\mathcal{A})^{4})(1+\log(1+\mathrm{diam}(\mathcal{A})))\nonumber
\end{equation}
This indicates that the convergence rate is closely related to the manifold diameter. For the distribution shift scenario, we need to consider the geometric characteristics of both the source manifold $\mathcal{A}$ and the target manifold $\widetilde{\mathcal{A}}$, so we can define the geometric factor:
\begin{equation}
\begin{aligned}
    \mathcal R(\mathcal{P}, \widetilde{\mathcal{P}}) =& C \cdot \left( \frac{\mathrm{diam}(\mathcal{A}) + \mathrm{diam}(\widetilde{\mathcal{A}})}{2} \right) \\&\cdot \left(1 + \log\left(1 + \mathbf{W}(\mathcal{P}, \widetilde{\mathcal{P}})\right)\right)
\end{aligned}\nonumber
\end{equation}
where $C$ incorporates constant terms such as dimension $d$ and diffusion parameters $\bar{\beta}$. 
\textcolor{black}{The multiplicative form of $\mathcal{R}$ reflects three considerations: the effective manifold diameter $(\mathrm{diam}(\mathcal{A}) + \mathrm{diam}(\widetilde{\mathcal{A}}))/2$ arises from the $W_1$-to-$L^2$ scaling on compact supports (since $\|x-y\|^2 \le D \cdot \|x-y\|$ on a bounded domain, the expected MSE is controlled by the joint diameter times $W_1$~\cite{villani09}); the logarithmic factor $1 + \log(1 + \mathbf{W})$ originates from the covering number estimates in the convergence theory of diffusion processes on compact manifolds~\cite{convergence}; and the multiplicative composition follows from the structure of the convergence bound and dimensional consistency.}

According to diffusion model convergence theory, the optimal performance of the generator is limited by its fit to the training distribution:
\begin{equation}
\mathbf{W}(G_\theta^*, \mathcal{P}) \leq \epsilon_{\text{opt}}\nonumber
\end{equation}
where $\epsilon_{\text{opt}}$ is the optimal convergence error of the diffusion model on the training distribution.

Therefore, the error term $\zeta$ satisfies:
\begin{equation}
|\zeta| \leq \mathcal R(\mathcal{P}, \widetilde{\mathcal{P}}) \cdot \mathbf{W}(G_\theta^*, \mathcal{P}) \leq \mathcal R(\mathcal{P}, \widetilde{\mathcal{P}}) \cdot \epsilon_{\text{opt}}\nonumber
\end{equation}

When $\mathcal{P}$ and $\widetilde{\mathcal{P}}$ are supported on the same manifold, $\mathrm{diam}(\mathcal{A}) = \mathrm{diam}(\widetilde{\mathcal{A}})$, the geometric factor simplifies to depend only on the diameter of a single manifold. When the distribution shift increases, $\mathbf{W}(\mathcal{P}, \widetilde{\mathcal{P}})$ increases, but the logarithmic term ensures a gentle growth of the geometric factor, which is consistent with the convergence behavior of diffusion models on manifolds.
\end{IEEEproof}

Theorem \ref{T:max} indicates that the minimum estimation error $E_{\text{min}}$ is primarily determined by the Wasserstein distance $\mathbf{W}(\mathcal{P}, \widetilde{\mathcal{P}})$ between the deployment distribution and the true underlying radio propagation law, and is scaled by the influence factor $\mathcal R(\mathcal{P}, \widetilde{\mathcal{P}})$. This factor captures the geometric differences in the underlying manifold structure of the data, and its theoretical form originates from the convergence bounds of diffusion models on compact manifolds. \textcolor{black}{Note that $E_{\text{min}}$ is a function of the random variable $\zeta$ that depends on the gap between the diffusion model with random noise as its input and the training distribution.}

\subsection{Sampling-Dependent Estimation Error}

Theorem \ref{T:max} characterizes \textcolor{black}{an express for the minimum estimation error} that can be achieved by a diffusion model in theory.
The estimation error can be positive, indicating a discrepancy between the estimated and ground-truth radio maps, due to potential mismatches between the deployment domain distribution and the true underlying radio propagation law.
As discussed in Section~\ref{sec:introduction}, the problem of radio map estimation is challenging due to ultra-low sampling rates, where the number of available measurements is far smaller than the target map resolution.
Intuitively, the quality of radio maps constructed by diffusion models improves as the input sampling rate increases.
However, a major limitation of Theorem \ref{T:max} is that the derived estimation error does not explicitly consider the sampling rate.

Define the estimation error $E$ \textcolor{black}{achieved by a diffusion model for radio map estimation} as
\begin{equation}
E=||\mathbf X-\widetilde{\mathbf M}||^2,
\end{equation}
where $\mathbf X$ is the map (matrix) predicted by a diffusion model.
Given a sampling mask $\Psi$, the amount of information about the complete matrix $\widetilde{\mathbf{M}}$ provided by the observation matrix $\mathbf{M}$ can be quantified by conditional mutual information:
\begin{equation}
\mathcal{I}(\Psi) = I(\widetilde{\mathbf{M}}; \mathbf{M} | \Psi) = H(\widetilde{\mathbf{M}}) - H(\widetilde{\mathbf{M}} | \mathbf{M}, \Psi)
\end{equation}
where $H(\cdot)$ denotes differential entropy. This metric reflects the reduction in uncertainty about the complete matrix $\widetilde{\mathbf{M}}$ given the observation $\mathbf{M}$ and the sampling pattern $\Psi$.

The efficacy of the generative model $G_\theta$ depends on its ability to effectively project observational information into the latent space $\mathcal{Z}$. Define the latent space information efficiency as:
\begin{equation}
\eta_z = \frac{I(\widetilde{\mathbf{M}}; G_\theta(\mathbf{N_Z}; \mathbf{M} \odot \Psi))}{I(\widetilde{\mathbf{M}}; \mathbf{M} | \Psi)} \cdot \frac{\dim(\mathcal{Z})}{\dim(\mathcal{M})}
\end{equation}
where $\dim(\mathcal{Z})$ and $\dim(\mathcal{M})$ denote the dimensions of the latent space and the matrix space, respectively. This efficiency factor comprehensively reflects the compression efficiency and information retention capability of the latent space.

Based on the information-theoretic framework, considering the influence of the observational information amount $\mathcal{I}(\Psi)$ and the latent space efficiency $\eta_z$ on the recovery error, we derive \textcolor{black}{an expression} for $E$ in the following theorem.

\begin{theorem}\label{thm:2}
Given the minimum estimation error $E_{min}$, the sampling mask $\Psi$, and the latent space efficiency $\eta_z$, the estimation error $E$ \textcolor{black}{can be approximated as}:
\begin{equation}\label{eqn:error}
\begin{aligned}
E ~=~& E_{min} + E_{info} + E_{model} \\
 \textcolor{black}{\approx} ~& \mathcal R(\mathcal{P}, \widetilde{\mathcal{P}})\cdot \mathbf{W}(\mathcal{P}, \widetilde{\mathcal{P}}) \\&+ \gamma \cdot \exp\left(-\lambda \cdot \frac{\mathcal{I}(\Psi)}{H(\widetilde{\mathbf{M}})} \cdot \eta_z^{\frac{\dim(\mathcal{M})}{\dim(\mathcal{Z})}}\right) + \epsilon
\end{aligned}
\end{equation}
where $\gamma > 0$ is the error sensitivity parameter, reflecting the amplification effect of missing observational information on the error; $\lambda > 0$ is the information decay coefficient, characterizing the marginal diminishing effect of information utilization; and $\epsilon \sim \mathcal{N}(0, \sigma^2)$ is the random disturbance term, encompassing uncertainties in model training and optimization.\label{T:mean}
\end{theorem}
\begin{IEEEproof}
    The estimation error $E$ can be decomposed into three interrelated components:
\begin{equation}
E = E_{\text{min}} + E_{\text{info}} + E_{\text{model}}\nonumber
\end{equation}
where $E_{\text{min}}$ is defined in Theorem \ref{T:max}, $E_{\text{info}}$ characterizes the error introduced by insufficient observational information (i.e., sampling strategy), and $E_{\text{model}}$ reflects the estimation error caused by imperfect latent space projection (i.e., model limitations).

According to rate-distortion theory~\cite{rateT}, under the condition of a given observational information amount $\mathcal{I}(\Psi)$, the minimum MSE achievable by the optimal estimation of the complete matrix $\widetilde{\mathbf{M}}$ satisfies:
\begin{equation}
\mathbb{E}[|\mathbf{X} - \widetilde{\mathbf{M}}|^2] \geq \sigma^2_{\widetilde{\mathbf{M}}} \cdot \exp\left(-2 \cdot \frac{\mathcal{I}(\Psi)}{H(\widetilde{\mathbf{M}})}\right)\nonumber
\end{equation}
where $\mathbf{X}$ is the estimation result, and $\sigma^2_{\widetilde{\mathbf{M}}}$ is the variance of $\widetilde{\mathbf{M}}$. However, due to the latent space bottleneck, the effective rate is reduced. Define the latent space information efficiency:
\begin{equation}
\eta_z = \frac{I(\widetilde{\mathbf{M}}; G_\theta(\mathbf{N_Z}; \mathbf{M} \odot \Psi))}{I(\widetilde{\mathbf{M}}; \mathbf{M} | \Psi)} \cdot \frac{\dim(\mathcal{Z})}{\dim(\mathcal{M})}\nonumber
\end{equation}
From the data processing inequality, $0 \leq \eta_z \leq 1$. The achievable effective rate becomes $R_{\text{eff}} = \eta_z^{\frac{d}{d_z}} \cdot \mathcal{I}(\Psi)$, where $d_z = \dim(\mathcal{Z})$, with the exponent arising from the concentration of measure phenomenon in information projection. Applying the rate-distortion bound with effective rate yields:
\begin{equation}
E_{info} ~\textcolor{black}{\approx} ~ \gamma \cdot \exp\left(-\lambda \cdot \frac{\mathcal{I}(\Psi)}{H(\widetilde{\mathbf{M}})} \cdot \eta_z^{\frac{\dim(\mathcal{M})}{\dim(\mathcal{Z})}}\right)\nonumber
\end{equation}
where normalizing by $H(\widetilde{\mathbf{M}})$ makes the argument dimensionless, and $\gamma, \lambda > 0$ are constants depending on the model architecture.

Third, the model approximation error $E_{model}$ arises from score estimation error and discretization. Under the assumptions of Section~\ref{sec:con}, the score estimation error satisfies:
\begin{equation}
\mathbb{E}_t \mathbb{E}_{\mathbf{x}(t)} [\|\boldsymbol{s}_{\boldsymbol{\theta}}(\mathbf{x}(t), t) - \nabla \log p_t(\mathbf{x}(t))\|^2] \leq \varepsilon_{\text{score}}^2\nonumber
\end{equation}
This error propagates to sample quality as $E_{model} \leq C_{\text{model}} \cdot \varepsilon_{\text{score}}^{2\alpha}$. In practice, we absorb this into $\eta_z$, with remaining stochasticity captured by $\epsilon \sim \mathcal{N}(0, \sigma^2)$.

Finally, by the bias-variance decomposition, the total MSE combines these components additively when errors are uncorrelated, which holds asymptotically for well-behaved estimators. Therefore, we obtain:
\begin{equation}
\begin{aligned}
E ~\textcolor{black}{\approx} ~& \mathcal R(\mathcal{P}, \widetilde{\mathcal{P}})\cdot \mathbf{W}(\mathcal{P}, \widetilde{\mathcal{P}}) \\&+ \gamma \cdot \exp\left(-\lambda \cdot \frac{\mathcal{I}(\Psi)}{H(\widetilde{\mathbf{M}})} \cdot \eta_z^{\frac{\dim(\mathcal{M})}{\dim(\mathcal{Z})}}\right) + \epsilon
\end{aligned}
\end{equation}
The exponential form for $E_{info}$ is justified by the rate-distortion function's exponential decay, while the power exponent $\frac{\dim(\mathcal{M})}{\dim(\mathcal{Z})}$ captures the geometric effect of dimensionality reduction on information preservation.

From the perspective of Wasserstein geometry, the Wasserstein distance between the output distribution of the generative model $G_\theta$ under the observational condition $\mathbf{M} \odot \Psi$ and the target distribution $\widetilde{\mathcal{P}}$ satisfies:
\begin{equation}
\mathbf{W}(G_\theta(\cdot; \mathbf{M} \odot \Psi), \widetilde{\mathcal{P}}) \leq \mathbf{W}(\mathcal{P}, \widetilde{\mathcal{P}}) + \mathbf{W}(G_\theta(\cdot; \mathbf{M} \odot \Psi), \mathcal{P})\nonumber
\end{equation}
The first term on the right-hand side corresponds to $E_{\text{min}}$, and the second term is jointly influenced by the observational information amount and model efficiency, exhibiting a consistent trend of variation with $E_{\text{info}} + E_{\text{model}}$.
\end{IEEEproof}

The proof of Theorem \ref{thm:2} is based on information projection theory and Wasserstein geometric properties.
Theorem \ref{thm:2} not only quantifies the approximation errors of diffusion-based radio map estimation but also provides guidance for the design of diffusion models and sampling strategies.

\subsection{Sampling Conditions to Achieve an Estimation Accuracy}

The estimation error $E$ can be decomposed into three components:
\begin{equation}\label{eqn:EFull}
E = E_{min} + E_{info} + E_{model}
\end{equation}
where $E_{min}$ is the minimum approximation error, $E_{info} = f(\mathcal{I}(\Psi))$ is the information deficiency error arising from information loss due to sparse sampling, and $E_{model} = g(\eta_z)$ is the model estimation error reflecting the imperfection of the latent space projection.

Considering that the functions $f(\cdot)$ and $g(\cdot)$ have the following asymptotic properties:
\begin{equation}
\lim_{\mathcal{I}(\Psi) \to H(\widetilde{\mathbf{M}})} f(\mathcal{I}(\Psi)) = 0, \quad \lim_{\eta_z \to 1} g(\eta_z) = 0
\end{equation}
and based on the error \textcolor{black}{expression} in Theorem \ref{T:mean}, we can derive conditions on radio sampling for diffusion models to achieve a given accuracy of radio map estimation.

Given an estimation error tolerance $\delta > 0$, \textcolor{black}{under random sampling conditions}, let $\rho_c$ be a sampling rate (density) such that when sampling rates $\rho > \rho_c$, the estimation error $E$ satisfies the following with a probability of at least $1 - \xi$:
\begin{equation}
\frac{E - E_{min}}{E_{min}} \leq \delta
\end{equation}
Then we have the following theorem for $\rho_c$.

\begin{theorem}\label{thm:convergence}
The sampling rate $\rho_c$ \textcolor{black}{can be approximated as}:
\begin{equation}\label{eqn:rhoc}
\begin{aligned}
    \rho_c \textcolor{black}{\approx} \min&\left\{\rho \in [0,1] : \gamma \cdot \exp\left(-\lambda \cdot \frac{\mathcal{I}(\Psi_\rho)}{H(\widetilde{\mathbf{M}})} \cdot \eta_z^{\frac{\dim(\mathcal{M})}{\dim(\mathcal{Z})}}\right) \leq \right.\\&\left. \delta \cdot E_{min} - \Phi^{-1}(1-\xi) \cdot \sigma_\epsilon\right\}
\end{aligned}
\end{equation}
where $\Psi_\rho$ denotes the sampling mask corresponding to $\rho$, and $\Phi^{-1}$ is the inverse cumulative distribution function of the standard normal distribution.
\end{theorem}
\begin{IEEEproof}
    Recall that Theorem \ref{T:mean} provides the following decomposition of the estimation error for a given sampling mask $\Psi$:
\begin{equation}
\begin{aligned}
E ~\textcolor{black}{\approx} ~& \mathcal R(\mathcal{P}, \widetilde{\mathcal{P}})\cdot \mathbf{W}(\mathcal{P}, \widetilde{\mathcal{P}}) \\&+ \gamma \cdot \exp\left(-\lambda \cdot \frac{\mathcal{I}(\Psi)}{H(\widetilde{\mathbf{M}})} \cdot \eta_z^{\frac{\dim(\mathcal{M})}{\dim(\mathcal{Z})}}\right) + \epsilon
\end{aligned}
\end{equation}
For a fixed sampling rate $\rho$, we consider a random sampling mask $\Psi_\rho$ drawn from the distribution corresponding to $\rho$. In the regime of interest, the mutual information $\mathcal{I}(\Psi_\rho)$ depends on $\rho$ in a deterministic manner (e.g., through its expectation under the random sampling scheme). Hence, we denote
\begin{equation}
A(\rho) \triangleq \frac{\mathcal{I}(\Psi_\rho)}{H(\widetilde{\mathbf{M}})} \cdot \eta_z^{\frac{\dim(\mathcal{M})}{\dim(\mathcal{Z})}},\nonumber
\end{equation}
which is a non‑decreasing function of $\rho$ because a higher sampling rate provides more information about the true radio map. The error term due to information deficiency then becomes $\Delta(\rho) = \gamma \exp\!\bigl(-\lambda A(\rho)\bigr)$, which is non‑increasing in $\rho$.

We aim to guarantee that, with probability at least $1-\xi$, the relative error satisfies
\begin{equation}
\frac{E - E_{\min}}{E_{\min}} \le \delta.\label{eq:Apx3:3}
\end{equation}
Substituting the expression for $E$, condition Eq. \eqref{eq:Apx3:3} is equivalent to
\begin{equation}
\Delta(\rho) + \varepsilon \le \delta E_{\min}.\label{eq:Apx3:4}
\end{equation}

Since $\varepsilon$ is Gaussian and independent of $\Psi_\rho$, for a fixed $\rho$ the probability of Eq. \eqref{eq:Apx3:4} is
\begin{equation}
\mathbb{P}\bigl(\Delta(\rho) + \varepsilon \le \delta E_{\min}\bigr) = \Phi\!\left(\frac{\delta E_{\min} - \Delta(\rho)}{\sigma_\varepsilon}\right),\label{eq:Apx3:5}
\end{equation}
where $\Phi$ denotes the cumulative distribution function of the standard normal distribution. To ensure that this probability is at least $1-\xi$, we require
\begin{equation}
\Phi\!\left(\frac{\delta E_{\min} - \Delta(\rho)}{\sigma_\varepsilon}\right) \ge 1-\xi.\label{eq:Apx3:6}
\end{equation}
Because $\Phi$ is strictly increasing, Eq. \eqref{eq:Apx3:6} is equivalent to
\begin{equation}
\frac{\delta E_{\min} - \Delta(\rho)}{\sigma_\varepsilon} \ge \Phi^{-1}(1-\xi),\label{eq:Apx3:7}
\end{equation}
where $\Phi^{-1}$ is the quantile function of the standard normal distribution. Rearranging gives
\begin{equation}
\Delta(\rho) \le \delta E_{\min} - \sigma_\varepsilon \Phi^{-1}(1-\xi).\label{eq:Apx3:8}
\end{equation}
Recalling the definition of $\Delta(\rho)$, inequality Eq. \eqref{eq:Apx3:8} becomes
\begin{equation}
\gamma \exp\!\left(-\lambda A(\rho)\right) \le \delta E_{\min} - \sigma_\varepsilon \Phi^{-1}(1-\xi).\label{eq:Apx3:9}
\end{equation}

The right-hand side of Eq. \eqref{eq:Apx3:9} is assumed to be positive for the threshold to be meaningful; otherwise no $\rho$ would satisfy the condition. Because $A(\rho)$ increases with $\rho$, the left‑hand side decreases with $\rho$. Therefore, the set of $\rho$ for which Eq. \eqref{eq:Apx3:9} holds is either empty or an interval $[\rho_c, 1]$, where $\rho_c$ is the smallest $\rho$ satisfying Eq. \eqref{eq:Apx3:9}. Formally,
\begin{equation}
\begin{aligned}
    \rho_c \textcolor{black}{\approx} \min&\left\{\rho \in [0,1] : \gamma \cdot \exp\left(-\lambda \cdot A(\rho)\right) \leq \right.\\&\left. \delta \cdot E_{min} - \Phi^{-1}(1-\xi) \cdot \sigma_\epsilon\right\}\nonumber
\end{aligned}
\end{equation}
Substituting back the definition of $A(\rho)$ yields exactly the expression in Theorem \ref{thm:convergence}. 
\textcolor{black}{Let $B=\delta E_{\min} - \sigma_\varepsilon \Phi^{-1}(1-\xi)$, where:
$$\begin{cases}
B \le 0: & \text{no sampling rate satisfies;} \\
B > \Delta(0): & \text{all sampling rates satisfy;} \\
\Delta(1) < B < \Delta(0): & \text{a threshold } \rho_c \text{ exists;} \\
0 < B < \Delta(1): & \text{even full sampling cannot guarantee.}
\end{cases}$$}
\end{IEEEproof}

{Theorem \ref{thm:convergence} extends Eq. \eqref{eq:rank} to the setting of diffusion-based radio map estimation}, and offers guidance for \textcolor{black}{the sampling rate required to achieve a given radio map estimation accuracy}. 

\section{Approximations to Theoretical Radio Map Estimation Errors of Diffusion Models}\label{sec:Empirical-formula}

In the previous section, for the problem of radio map estimation, we derived \textcolor{black}{expressions for} the minimum estimation error that can be achieved by diffusion models (Theorem~\ref{T:max}), and further extended this error to incorporate the sampling strategy (Theorem~\ref{T:mean}).
\textcolor{black}{In addition, based on the sampling-dependent estimation error, we also derived conditions on sampling for diffusion models to achieve a given estimation accuracy (Theorem~\ref{thm:convergence})}. 
However, direct application of these \textcolor{black}{expressions} in practical scenarios is hindered by their reliance on information that is often unavailable or difficult to obtain. 
To address this limitation, in this section, we introduce empirical formulas which approximate the theoretical error \textcolor{black}{expressions} and are readily computable in practice.

\subsection{Approximation to Minimum Estimation Error}\label{subsec:EminStar}

For the minimum estimation error $E_{\min}$, $\mathrm{diam}(\mathcal{A})$ (diameter of the manifold $\mathcal{A}$), $\mathrm{diam}(\widetilde{\mathcal{A}})$ (diameter of the manifold $\widetilde{\mathcal{A}}$), and $\mathbf{W}(\mathcal{P}, \widetilde{\mathcal{P}})$ (Wasserstein distance between distributions $\mathcal{P}$ and $\widetilde{\mathcal{P}}$) are difficult to calculate in practice.
However, since datasets are empirical samples drawn from the underlying radio map distributions, they can serve as practical proxies for this distributional information.
Leveraging principles from probability theory and statistics, we therefore approximate $\mathrm{diam}(\mathcal{A})$, $\mathrm{diam}(\widetilde{\mathcal{A}})$, and $\mathbf{W}(\mathcal{P}, \widetilde{\mathcal{P}})$ using observable datasets and their empirical similarity. 
The proposed approximations are formalized in the following corollary.

\begin{corollary}\label{cor:1} 
The following \( E_{\text{min}}^* \) approximates the minimum estimation error \( E_{\text{min}} \) in Eq.~\eqref{eqn:EMIN}:
\begin{equation}\label{eqn:eminstar}
\begin{aligned}
    E_{\text{min}}^* ~=~& C' \cdot \left( \frac{\text{diam}(Y) + \text{diam}(\widetilde{Y})}{2} \right)\\& \cdot \left(1 + \log\left(1 + D(Y, \widetilde{Y})\right)\right) \cdot D(Y, \widetilde{Y}) + \widetilde{\zeta}
\end{aligned}
\end{equation}
where \( D(Y, \widetilde{Y}) \) is a distribution difference metric computed based on the datasets \( Y \) and \( \widetilde{Y} \), used to approximate the theoretical Wasserstein distance and can be measured by methods such as Fréchet Inception Distance (FID)~\cite{FID} and Maximum Mean Discrepancy (MMD)~\cite{MMD}. Here, \( Y \) and \( \widetilde{Y} \) are finite datasets from the target distribution and the training distribution, respectively; \( \mathrm{diam}(\cdot) \) denotes the empirical diameter of a dataset (e.g., the maximum distance between samples); \( C' \) is a fitting constant related to the model and task. \( \widetilde{\zeta} \) is an error term reflecting the empirical approximation error of the diffusion model \( G_\theta \) on its training distribution \( \widetilde{Y} \), which can be estimated through the model's performance on a held-out test set from the training distribution.\label{C:EFTME}
\end{corollary}

The empirical formula in Eq.~\eqref{eqn:eminstar} provides a practical approximation of the minimum estimation error in Eq.~\eqref{eqn:EMIN}.
Specifically, it uses the maximum distance between samples in the dataset $Y$ to approximate the diameter of the manifold $\mathcal{A}$, uses the maximum distance between samples in the dataset $\widetilde{Y}$ to approximate the diameter of the manifold $\widetilde{\mathcal{A}}$, and uses FID or MMD to approximate the Wasserstein distance between $\mathcal{P}$ and $\widetilde{\mathcal{P}}$.
Unlike its theoretical counterpart, Eq. \eqref{eqn:eminstar} is readily computable in practice and thus well-suited for real-world applications.
Corollary~\ref{cor:1} suggests that our proposed approximation to the minimum estimation error is related to the mismatch between the training and deployment environments.
\textcolor{black}{Similar to $E_{\text{min}}$ defined in Eq.~\eqref{eqn:EMIN}, $E_{\text{min}}^*$ is a function of the random variable $\widetilde{\zeta}$ that depends on the gap between the diffusion model with random noise as its input and the training dataset.}

\subsection{Approximation to Sampling-Dependent Estimation Error}

For the estimation error $E$, information-related metrics, e.g., the mutual information $\mathcal{I}(\Psi)$ and the differential entropy $H(\widetilde{\mathbf{M}})$, are difficult to calculate in practice.
In the following corollary, we provide an empirical formula to approximate $E$.

\begin{corollary}\label{cor:2}
Considering that the layout of obstacles and \textcolor{black}{the degree of radio power variation near sampling points} significantly impact the estimation of radio maps, we propose to use the following $E^*$ to approximate the sampling-dependent estimation error $E$ in Eq.~\eqref{eqn:error}:
\textcolor{black}{
\begin{equation}\label{eqn:estar}
E^* = \max\{ E^*_{min} + E^*_{model} - \alpha\ln(\rho^*+1),\ 0\}
\end{equation}}
with
\begin{equation}
\rho^* = \frac{c_1 \cdot (\kappa_{sp} + c_2) \cdot (\rho + c_3)}{\exp\left[{\frac{1}{mn} \sum_{i=1}^{m} \sum_{j=1}^{n} \mathbf{B}_{ij}}\right]}
\end{equation}
where \textcolor{black}{\( \rho \) represents the actual sampling rate}, \( \mathbf{B} \) is the obstacle layout mask, with 1 indicating the locations of the obstacles and 0 for others, \( \alpha, c_1, c_2, c_3 \) are normalization hyperparameters, and \( \kappa_{sp} \) represents the degree of power variation at sampling points and can be computed as:
\begin{equation}
\kappa_{sp} = \frac{\Psi \odot \sqrt{ (\nabla_x \mathbf{X})^2 + (\nabla_y \mathbf{X})^2 }}{\sum_{i=1}^{m}\sum_{j=1}^{n} \Psi_{ij}}
\end{equation} \label{C:EFIDE}
Let \( \mathbf{X}_{\text{true}} \) be the true matrix of a radio map, and \( \Psi_{\text{train}} \) and \( \Psi_{\text{test}} \) be the training and test sampling masks, respectively.
Then we have
\begin{equation}\label{eqn:emodelstar}
\begin{aligned}
E^*_{model}
~&~= \mathbf{W}(G_{\theta^*}(\cdot; \Psi_{\text{train}}), \mathcal{P}_{\text{train}}) \\
&\quad + \mathbf{W}(\mathcal{P}_{\text{train}}, \mathcal{P}_{\text{test}}), \\
&\quad + \mathbf{W}(G_{\theta^*}(\cdot; \Psi_{\text{test}}), G_{\theta^*}(\cdot; \Psi_{\text{train}}))
\end{aligned}
\end{equation}
where \( \mathcal{P}_{\text{train}} = \mathcal{P}(\cdot|\mathbf{X}_{\text{true}} \odot \Psi_{\text{train}}) \) and  \( \mathcal{P}_{\text{test}} = \mathcal{P}(\cdot|\mathbf{X}_{\text{true}} \odot \Psi_{\text{test}}) \). For brevity, denote \( G_\theta(\cdot; \Psi) \equiv G_\theta(\mathbf{N_Z}; \mathbf{X}_{\text{true}} \odot \Psi) \).
In addition, let \( \theta^* \) denote the model parameters optimized on the training dataset, aiming to minimize the training distribution matching error and the observation reconstruction error:
\begin{equation}
\begin{aligned}
    \theta^* =& \arg\min_{\theta} \mathbf{W}(G_\theta(\cdot; \Psi_{\text{train}}), \mathcal{P}_{\text{train}}) \\&+ \lambda \|(G_\theta(\cdot; \Psi_{\text{train}}) - \mathbf{X}_{\text{true}}) \odot \Psi_{\text{train}}\|_F^2.
\end{aligned}
\end{equation}
\end{corollary}

The empirical formula in Eq.~\eqref{eqn:estar} provides a practical approximation of the sampling-dependent estimation error in Eq.~\eqref{eqn:error}.
Specifically, as shown in Eq.~\eqref{eqn:EFull}, the sampling-dependent estimation error consists of the minimum estimation error (defined in Eq.~\eqref{eqn:EMIN}), the information deficiency error (arising from information loss due to sparse sampling), and the model estimation error (reflecting the imperfection of the latent space projection).
Eq.~\eqref{eqn:estar} uses Eq.~\eqref{eqn:eminstar} in Corollary~\ref{cor:1} to approximate the minimum estimation error, leverages the layout of obstacles and the degree of radio power variation near sampling points, both of which significantly impact the estimation of radio maps, to approximate the information deficiency error, and employs the differences between training and test datasets to approximate the model estimation error.

\textcolor{black}{The empirical formula in Eq.~\eqref{eqn:estar} provides a function relating the estimation error of diffusion models to the sampling rate, and it is readily computable from observable data when only a limited number of radio samples is available and ground-truth radio maps are unavailable. 
Intuitively, this empirical formula allows one not only to calculate the estimation error for a given sampling rate but also to determine the sampling rate required to achieve a prescribed estimation accuracy. 
Consequently, from the empirical formula in Eq.~\eqref{eqn:estar}, one can also readily derive an empirical formula for the critical sampling rate $\rho_c$ defined in Eq.~\eqref{eqn:rhoc}.}

\subsection{Training Sampling Rate vs Test Sampling Rate}

In the context of using diffusion models for radio map estimation, the sampling rate employed during model training may differ from that used during model testing.
Intuitively, the training sampling rate should be no smaller than the test sampling rate, as the model is optimized on the training sampling rate and may not generalize well to higher-rate inputs unseen during training.
However, \emph{to improve the test-time performance of a diffusion model for radio map estimation, should the test-time sampling rate be equal to, or less than the training-time sampling rate?}
In this section, we address this question by analyzing the estimation error component which depends on the model design and is derived in Eq.~\eqref{eqn:emodelstar}.

Let the sampling rate during training be \( \rho_{\text{train}} \) and that during test be \( \rho_{\text{test}} \) with \( \rho_{\text{train}}\geq\rho_{\text{test}} \), then we have the following corollary.

\begin{corollary}\label{cor:3}
\textcolor{black}{Given a test sampling rate $\rho_{\text{test}}$, there exists a training sampling rate \( \rho_{\text{train}}^* > \rho_{\text{test}} \) that minimizes Eq.~\eqref{eqn:emodelstar}.\label{C:EFMAE}}
\end{corollary}

Based on Eq.~\eqref{eqn:emodelstar}, consider
\begin{equation}
\begin{aligned}
E^*_{model}
~&~= \underbrace{\mathbf{W}(G_{\theta^*}(\cdot; \Psi_{\text{train}}), \mathcal{P}_{\text{train}})}_{(A)} \\
&\quad + \underbrace{\mathbf{W}(\mathcal{P}_{\text{train}}, \mathcal{P}_{\text{test}})}_{(B)}, \\
&\quad + \underbrace{\mathbf{W}(G_{\theta^*}(\cdot; \Psi_{\text{test}}), G_{\theta^*}(\cdot; \Psi_{\text{train}}))}_{(C)}
\end{aligned}
\end{equation}
In the following, we analyze each term and derive Corollary~\ref{cor:3}. 

\smallskip
\noindent \textbf{Training Distribution Matching Error (A)} \quad
This term characterizes the model's fit to the training conditional distribution. Since \( \rho_{\text{train}} \geq \rho_{\text{test}} \), the training phase provides more observational information, enabling the model to better learn the underlying radio propagation laws. 
Therefore, assuming sufficient model capacity, (A) monotonically decreases as \( \rho_{\text{train}} \) increases. Denote (A) as \( f(\rho_{\text{train}}) \), where \( f(\cdot) \) is a decreasing function.

\smallskip
\noindent \textbf{Conditional Distribution Shift Term (B)} \quad
This term reflects the difference in true distributions caused by differing observation conditions. Based on the physical smoothness and local correlation of radio propagation, there exists a constant \( L>0 \) such that the Wasserstein distance between two conditional distributions satisfies:
    \begin{equation}
    \mathbf{W}(\mathcal{P}(\cdot|\Psi_1), \mathcal{P}(\cdot|\Psi_2)) \leq L \cdot \mathbb{E}_{(i,j)}[|\Psi_1(i,j) - \Psi_2(i,j)|].
    \end{equation}
    When \( \Psi_{\text{train}} \) and \( \Psi_{\text{test}} \) are sampled i.i.d., the expected difference is \( |\rho_{\text{train}} - \rho_{\text{test}}| \), thus \( (B) \leq L \cdot |\rho_{\text{train}} - \rho_{\text{test}}| \).

\smallskip
\noindent \textbf{Conditional Input Difference Term (C)} \quad
This term measures the sensitivity of the model output to changes in the input mask. 
For a sufficiently trained diffusion model that robustly understands radio propagation laws, its output distribution should remain stable to small perturbations in the input. Assuming that the model satisfies Lipschitz continuity, there exists a constant \( K>0 \) such that:
    \begin{equation}
    \mathbf{W}(G_{\theta^*}(\cdot; \Psi_1), G_{\theta^*}(\cdot; \Psi_2)) \leq K \cdot \mathbb{E}_{(i,j)}[|\Psi_1(i,j) - \Psi_2(i,j)|],
    \end{equation}
    thus \( (C) \leq K \cdot |\rho_{\text{train}} - \rho_{\text{test}}| \).

Combining the above analysis, the model estimation error is upper bounded by:
\begin{equation}
\mathbb{E}[\mathcal{E}_{\text{approx}}] \leq f(\rho_{\text{train}}) + (L + K) \cdot (\rho_{\text{train}} - \rho_{\text{test}}).
\end{equation}
When the test sampling rate \( \rho_{\text{test}} \) is fixed, the upper bound $f(\rho_{\text{train}}) + (L + K) \cdot (\rho_{\text{train}} - \rho_{\text{test}})$ first decreases and then increases with \( \rho_{\text{train}} \) (\( \rho_{\text{train}} \geq \rho_{\text{test}} \)): initially, when \( \rho_{\text{train}} \) increases , the rapid decrease of \( f(\rho_{\text{train}}) \) dominates, causing the upper bound to decrease; as \( \rho_{\text{train}} \) increases further, the growth of the linear term \( (L+K)(\rho_{\text{train}}-\rho_{\text{test}}) \) gradually dominates, leading to an increase in the upper bound. Therefore, there exists an optimal \( \rho_{\text{train}}^* > \rho_{\text{test}} \) that minimizes the upper bound of \( \mathbb{E}[\mathcal{E}_{\text{approx}}] \). 

Corollary~\ref{cor:3} establishes that, for diffusion-based radio map estimation, setting the training sampling rate slightly above the test sampling rate achieves an optimal trade-off. 
This configuration balances two competing effects: it provides sufficient training data for the model to learn the underlying radio propagation laws and effectively complete missing information, while avoiding excessive conditional mismatch that would degrade performance in low-sampling test scenarios.
Specifically, when \( \rho_{\text{train}} = \rho_{\text{test}} \), the model may lack the capacity to generalize beyond the observed samples; when \( \rho_{\text{train}} \gg \rho_{\text{test}} \), the disparity in sampling conditions impairs the model's ability to perform accurately under the test-time sampling regime.

\section{Numerical Results}\label{sec:evaluation}

In this section, we use large-scale simulations to evaluate our proposed empirical formulas which approximate our \textcolor{black}{derived expressions for estimation errors achieved} by diffusion models for radio map estimation.

\subsection{Simulation Setup}

We use the open-source radio map dataset \emph{BART-Lab radio maps}\footnote{https://github.com/BRATLab-UCD/Radiomap-Data} from~\cite{DataSet} for evaluation.
This dataset is obtained by the simulation software \emph{Altair Feko} from \emph{WinProp}~\cite{altair_winprop_2021} and consists of 2000 radio maps.
The number of transmitters in one radio map varies from $1$ to $3$.
There are 2000 different layouts of buildings in the dataset, extracted from \emph{OpenStreetMap}~\cite{openstreetmap} in the United States.
The heights of the buildings are set as $10$ m.
This dataset consists of radio maps simulated in 5 different frequency bands: 1750 MHz, 2750 MHz, 3750 MHz, 4750 MHz, and 5750 MHz. 
We select the 5750 MHz subset for evaluation.
Considering that the resolutions of the original maps in the dataset are different, to ensure consistency, we randomly extract $256 \times 256$ (i.e., $m=n=256$) regions from each map.
This results in 1776 different radio maps, each of which has a resolution of $256 \times 256$.
Among them, $1676$ are used for model training and the remaining $100$ are used for test.
During test, the sampling rate varies from $1/(256\times 256)=0.002\%$ to $300/(256\times 256)=0.458\%$.
For a given sampling rate, radio samples are \textcolor{black}{randomly} distributed within the region of interest.

We simulate \emph{WiFi-Diffusion} from~\cite{wifi-diffusion} to evaluate our proposed empirical formulas. 
\emph{WiFi-Diffusion} first uses DDPM, a well-known diffusion model, to generate a candidate set of radio maps ($64$ radio maps in simulations) with diverse qualities under ultra-low sampling rates, and then employs principles of radio propagation to identify the most accurate map.
\textcolor{black}{We implement \emph{WiFi-Diffusion} on an NVIDIA RTX 4090 GPU and an AMD EPYC 9654 CPU. It is trained for 44 GPU hours.}
Deep learning codes are built using PyTorch.
\textcolor{black}{For the hyperparameters in the empirical formulas, preliminary experiments should be conducted in the target environment before the formulas are applied. 
By calibrating the hyperparameters using the preliminary results, one can obtain a configuration of the empirical formulas suitable for the target environment.}

\subsection{Simulation Results}

\begin{figure}[]
	\centering
	\includegraphics[width=\linewidth]{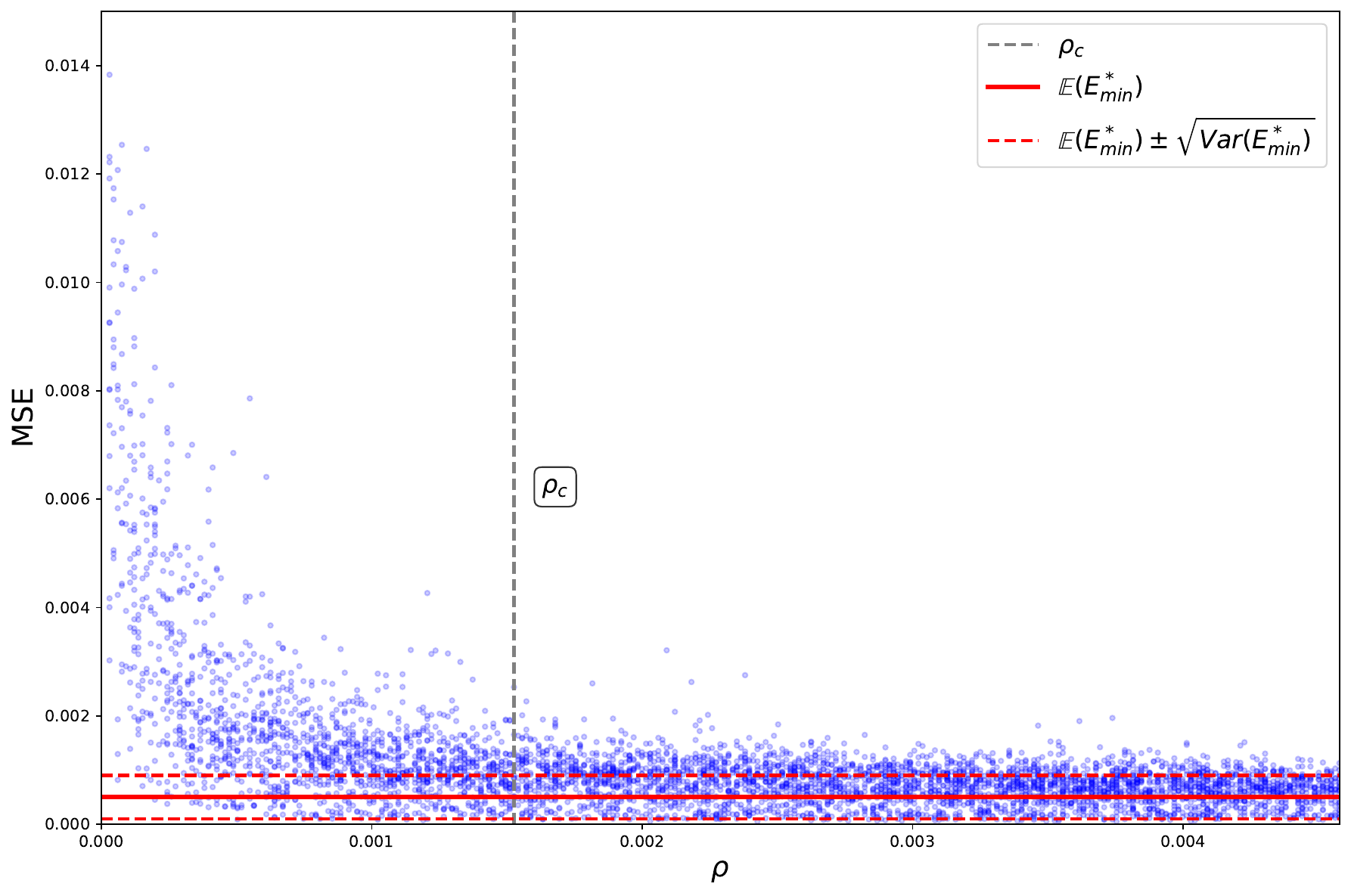}
    \caption{Simulation results to evaluate Corollary~\ref{cor:1}.}
	\label{fig:demo}
\end{figure}

\smallskip
\noindent \textbf{Evaluation on Corollary~\ref{cor:1}} \quad
In Theorem~\ref{T:max}, we derived \textcolor{black}{a theoretical expression for} $E_{\min}$, the minimum estimation error that can be achieved by a diffusion model for radio map estimation.
Since direct computation of this expression is difficult in practice, in Corollary~\ref{cor:1}, we introduced an empirical approximation $E_{\min}^*$ that is readily computable.
To assess the fidelity of this approximation, we conduct $5000$ simulation instances, each with a sampling rate $\rho$ randomly drawn from the interval $[0.002\%, 0.458\%]$. 
Fig.~\ref{fig:demo} illustrates the resulting MSE for each instance.

Fig.~\ref{fig:demo} reveals variability in the quality of maps produced by \emph{WiFi-Diffusion}, attributable to the stochastic nature of diffusion models (the dependence of the diffusion model's output on a random input noise).
Moreover, in the figure, the solid red line represents the expectation of $E_{\min}^*$ (recall that as discussed in Section~\ref{subsec:EminStar}, $E_{\min}^*$ is a function of the random variable $\widetilde{\zeta}$), while the dashed red lines indicate the standard deviation of $E_{\min}^*$.
Comparing these red lines with the MSE values of individual instances, we observe that $E_{\min}^*$ is an effective approximation to the minimum estimation error, as the majority of MSE values lie close to the expectation and within one standard deviation, with only minor deviations.

\textcolor{black}{Fig.~\ref{fig:demo} also suggests that the estimation performance of diffusion models for radio map estimation converges as the sampling rate increases. In fact, Theorem~\ref{thm:convergence} derives the sampling rate $\rho_c$ required to ensure that the estimation error is within $\delta$ of the minimum error $E_{\min}$, thereby providing guidance for determining the sampling rate threshold for performance convergence. 
In Fig.~\ref{fig:demo}, we also plot $\rho_c$ for $\delta = 0.4$. 
We observe that the quality of radio maps constructed by \emph{WiFi-Diffusion} remains relatively stable as the sampling rate exceeds $\rho_c$.
}

\begin{figure}[]
	\centering
	\includegraphics[width=\linewidth]{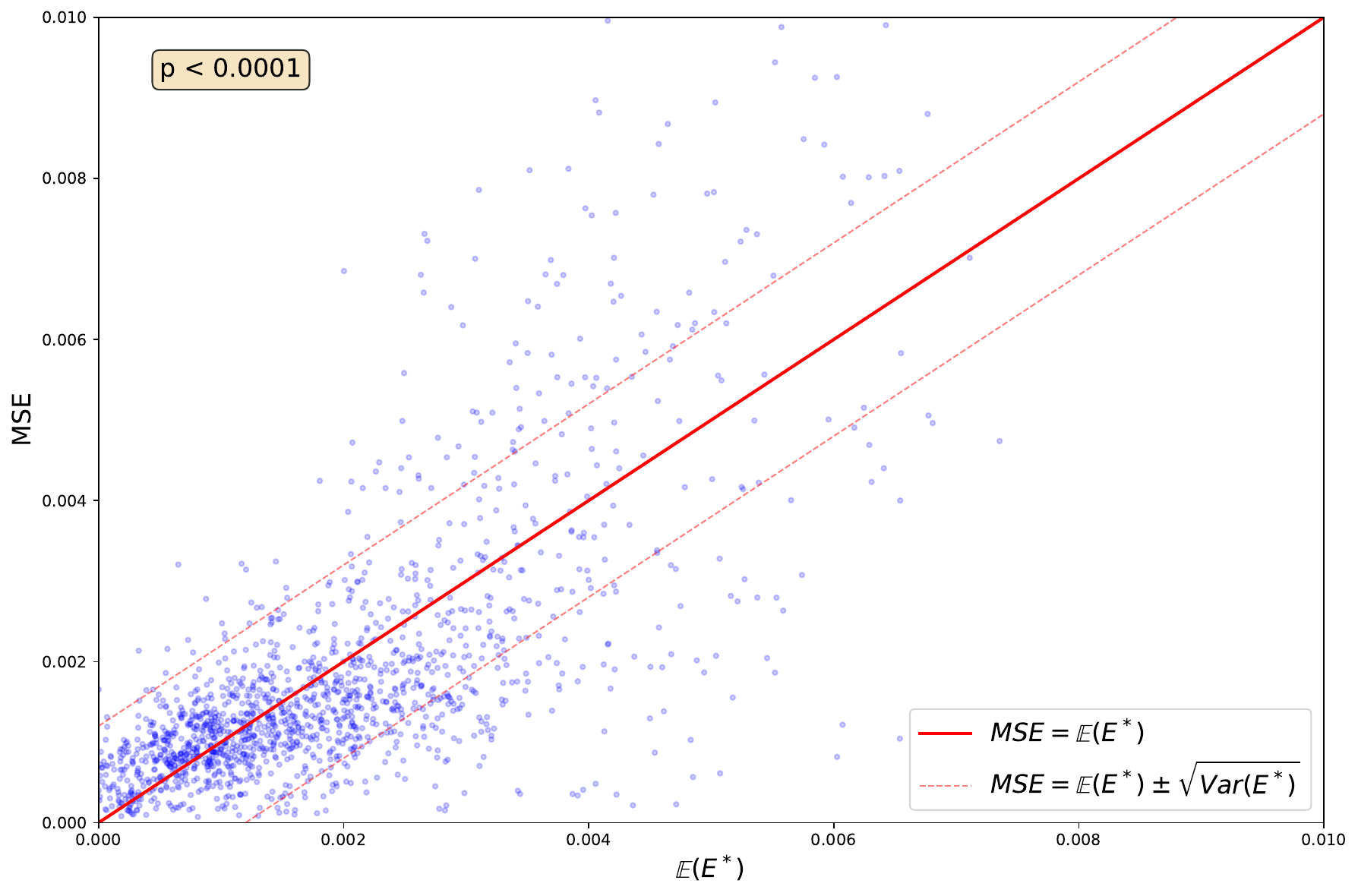}
    \caption{Simulation results to evaluate Corollary~\ref{cor:2}.}
	\label{fig:IDE}
\end{figure}

\smallskip
\noindent \textbf{Evaluation on Corollary~\ref{cor:2}} \quad
In Theorem~\ref{T:mean}, we derived \textcolor{black}{a theoretical expression for} $E$, the sampling-dependent estimation error that can be achieved by a diffusion model for radio map estimation.
Since direct computation of $E$ is difficult in practice, in Corollary~\ref{cor:2}, we introduced an empirical approximation $E^*$ that is readily computable.
To validate the effectiveness of this approximation, we conducted simulations and present the results in Fig.~\ref{fig:IDE}.
Specifically, for each simulated instance with a sampling rate below $0.0001$, we plot the MSE of the reconstructed radio map (y-axis) against the corresponding expectation of $E^*$ (x-axis).
We focus on low sampling rates because the estimation error $E$ is most sensitive to sampling density in this regime, i.e., the deviation between $E$ and the minimum error $E_{\min}$ becomes pronounced.
The figure confirms that $E^*$ serves as a reliable approximation of $E$: most data points are near the diagonal line, indicating close agreement between the achieved MSE and the empirical approximation for the sampling-dependent estimation error.

\begin{figure}[]
	\centering
\includegraphics[width=\linewidth]{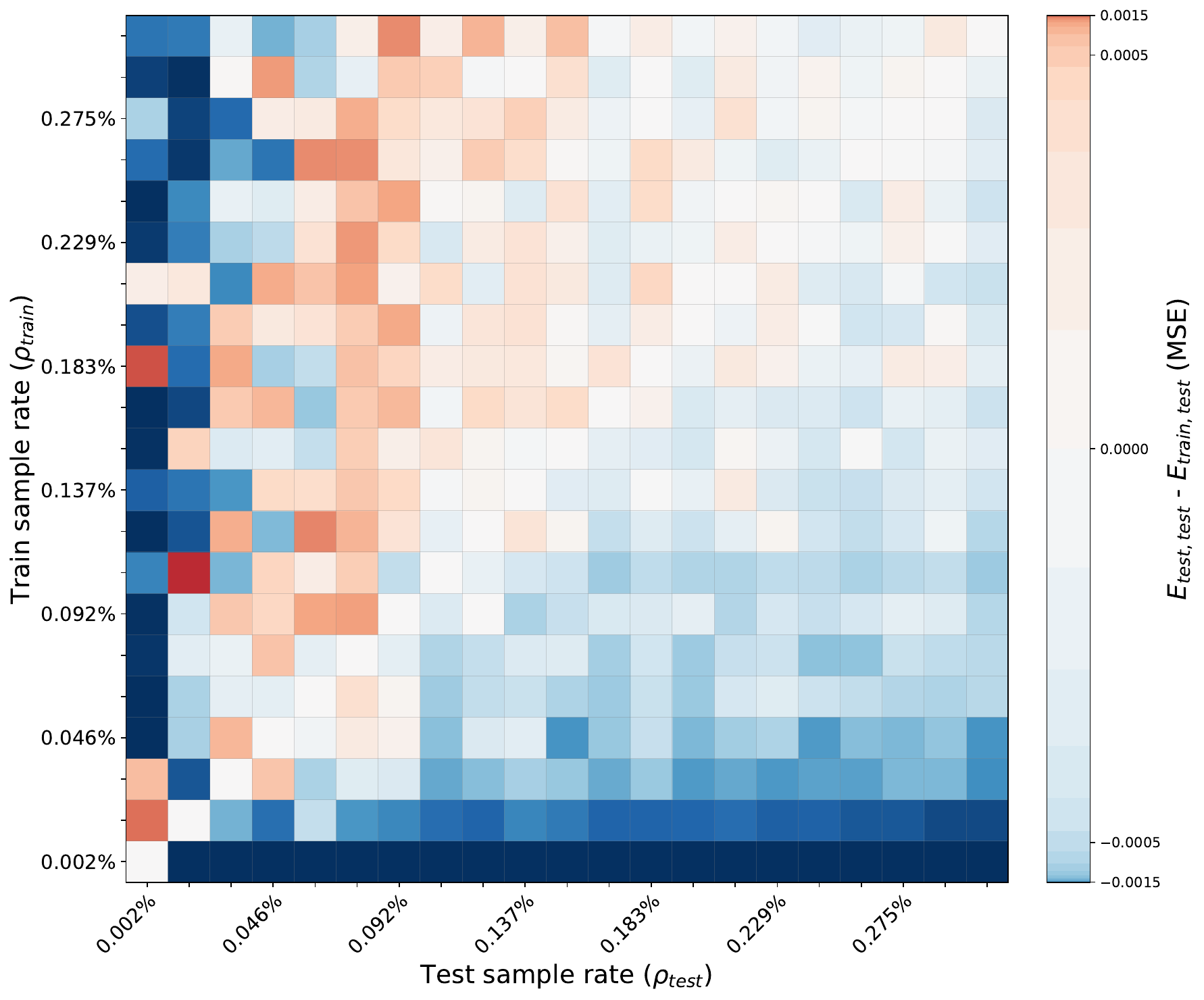}
    \caption{Simulation results to evaluate Corollary~\ref{cor:3}.}
	\label{fig:COR3}
\end{figure}

\smallskip
\noindent \textbf{Evaluation on Corollary~\ref{cor:3}} \quad
In Corollary~\ref{cor:3}, we established that for diffusion models, given a fixed test-time sampling rate, setting the training sampling rate slightly above the test sampling rate minimizes the estimation error at test time.
We now present simulation results to validate this corollary.
Consider the following set of sampling rates $P$:
\begin{equation}
\begin{aligned}
    P=\{&0.002\%,0.015\%,0.031\%,0.046\%,0.061\%,0.076\%,\\
    &0.092\%,0.107\%,0.122\%,0.137\%,0.153\%,0.168\%,\\
    &0.183\%,0.198\%,0.214\%,0.229\%,0.244\%,0.259\%,\\
    &0.275\%,0.290\%,0.305\%\}.\nonumber
\end{aligned}
\end{equation}
For each $\rho_{\text{train}}\in P$, we trained a diffusion model using that sampling rate.
Then each trained model was evaluated on every $\rho_{\text{test}}\in P$, yielding a total of $21\times 21$ MSE values corresponding to all pairs $(\rho_{\text{test}},\rho_{\text{train}})$.
The results are visualized in Fig.~\ref{fig:COR3}.

In Fig.~\ref{fig:COR3}, each cell indexed by $(\rho_{\text{test}},\rho_{\text{train}})$ displays the MSE difference $E_{\text{test},\text{test}}-E_{\text{train},\text{test}}$, where $E_{\text{test},\text{test}}$ denotes the MSE of a model trained and tested at the same sampling rate $\rho_{\text{test}}$, while $E_{\text{train},\text{test}}$ denotes the MSE of a model trained at $\rho_{\text{train}}$ and tested at $\rho_{\text{test}}$.
A larger value of this difference (indicated by a redder cell in the figure) suggests that the model trained at $\rho_{\text{train}}$ outperforms the baseline model trained and tested at $\rho_{\text{test}}$.
From Fig.~\ref{fig:COR3}, we observe that for each $\rho_{\text{test}}\in P$, the maximum difference occurs in the region above the diagonal, i.e., where $\rho_{\text{train}}>\rho_{\text{test}}$.
This confirms that the optimal training sampling rate for a given test rate exceeds the test rate itself, thereby validating Corollary~\ref{cor:3}.  

\begin{figure}[]
	\centering
	\includegraphics[width=\linewidth]{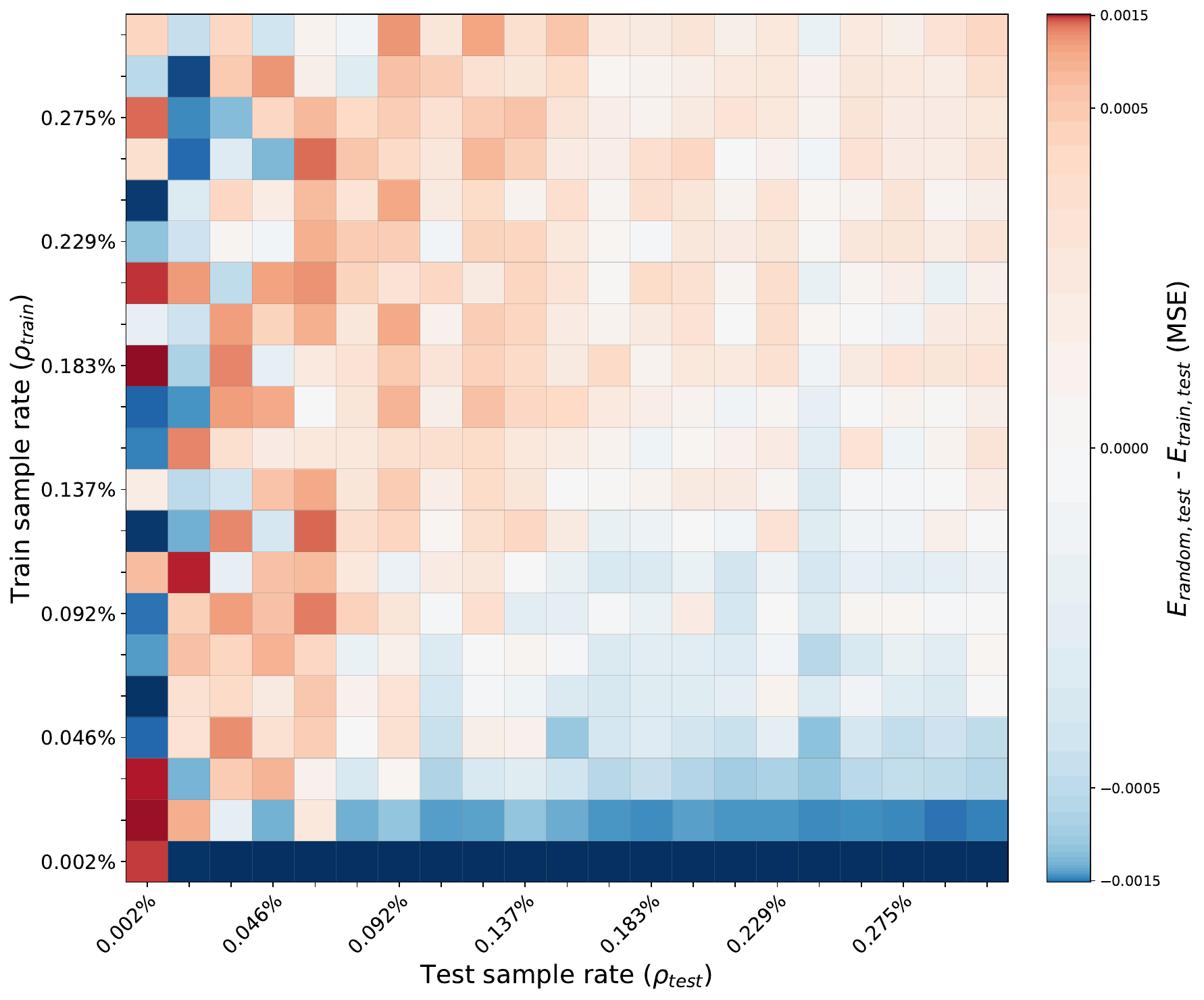}
    \caption{Simulation results comparing diffusion models trained with random sampling rates with those trained with a fixed sampling rate.}
	\label{fig:train}
\end{figure}

\smallskip
\noindent \textbf{More Simulation Results} \quad
We note that diffusion models can also be trained using random sampling rates, where each training instance employs a potentially different sampling rate. 
This raises an interesting question: Does training with random sampling rates yield better performance than training with a fixed sampling rate? 
Finally, we conducted simulations to address this question.

In Fig.~\ref{fig:train}, each cell indexed by $(\rho_{\text{test}},\rho_{\text{train}})$ displays the MSE difference $E_{\text{random},\text{test}}-E_{\text{train},\text{test}}$, where $E_{\text{random},\text{test}}$ denotes the MSE of a model trained at random sampling rates (i.e., for each training instance, the sampling rate is drawn uniformly from $P$) and tested at $\rho_{\text{test}}$, while $E_{\text{train},\text{test}}$ denotes the MSE of a model trained at a fixed rate $\rho_{\text{train}}$ and tested at $\rho_{\text{test}}$.
A positive value of this difference (indicated by a warmer color in the figure) indicates that the fixed-rate model outperforms the random-rate model.
From the figure, we observe that for each $\rho_{\text{test}}\in P$, the maximum difference occurs in the region above the diagonal.
This result demonstrates that training with a fixed sampling rate, particularly one slightly above the test rate, yields superior radio map estimation performance compared to training with random sampling rates.

\section{Conclusion}\label{sec:conclusion}
The radio map estimation problem requires one to construct a radio map from limited radio samples collected by sensors sparsely distributed within a region of interest.
It is challenging due to the extreme sparsity of radio samples relative to map resolution. 
While diffusion models have recently emerged as a promising approach for this problem, their \textcolor{black}{performance is typically assessed only by comparing the estimated radio maps with ground-truth radio maps, and their} theoretical performance remains largely unexplored.
This work provides the missing theoretical foundation.

Specifically, by formulating radio map estimation as a nonlinear matrix completion problem, we \textcolor{black}{derived a theoretical expression for the estimation error of a diffusion model, capturing the effects of the mismatch between the training and deployment environments, the model design, and the sampling strategy on the quality of the estimated radio map.}
Moreover, as the derived error \textcolor{black}{expression} relies on certain information difficult to obtain in practice, we proposed an empirical formula which approximates the error expression and is readily computable from observable data.
Extensive simulations verified that our proposed approximations are tight.

\textcolor{black}{Our results provide guidelines not only for evaluating the performance of a diffusion model when ground truth is unavailable, but also for determining the sampling rate required to achieve a given estimation accuracy. 
They further reveal a sampling rate threshold above which the estimation performance of diffusion models converges.}
As future research directions, we plan to design diffusion-based radio map estimation models that surpass current approaches and approach the theoretical accuracy \textcolor{black}{expressions} established in this work, and develop radio sampling strategies adaptive to dynamic environments, enhancing the estimation performance of diffusion models.

\bibliographystyle{IEEEtran}
\bibliography{Theory}

\end{document}